\documentclass[a4paper,english,11pt]{article}

\usepackage{libertineRoman}
\usepackage{geometry}
\usepackage[T1]{fontenc}
\usepackage[utf8]{inputenc}

\usepackage{calc}
\usepackage{amsmath}
\usepackage{amsthm}
\usepackage{amssymb}
\usepackage{graphicx}
\usepackage{nameref}
\usepackage{multicol}
\usepackage{sectsty}
\allsectionsfont{\boldmath}

\makeatletter

\providecommand{\tabularnewline}{\\}

\numberwithin{equation}{section}
\numberwithin{figure}{section}
\theoremstyle{plain}
\newtheorem{thm}{\protect\theoremname}
\theoremstyle{plain}
\newtheorem{lem}[thm]{\protect\lemmaname}
\theoremstyle{remark}
\newtheorem{rem}[thm]{\protect\remarkname}
\theoremstyle{definition}
\newtheorem{defn}[thm]{\protect\definitionname}
\theoremstyle{plain}
\newtheorem{prop}[thm]{\protect\propositionname}
\theoremstyle{remark}
\newtheorem*{rem*}{\protect\remarkname}
\theoremstyle{plain}
\newtheorem{cor}[thm]{\protect\corollaryname}

\date{}
\usepackage{float}
\usepackage[framemethod=tikz]{mdframed}

\mdfdefinestyle{myFigureBoxStyle}{tikzsetting={draw=black, line width=1pt}}%

\makeatletter
  \newcommand\fs@myRoundBox{\def\@fs@cfont{\bfseries}\let\@fs@capt\floatc@plain
  \def\@fs@pre{\begin{mdframed}[style=myFigureBoxStyle]}%
  \def\@fs@mid{\vspace{\abovecaptionskip}}%
  \def\@fs@post{\end{mdframed}}\let\@fs@iftopcapt\iffalse}
\makeatother

\floatstyle{myRoundBox}
\restylefloat{figure}

\makeatother

\usepackage{babel}
\providecommand{\corollaryname}{Corollary}
\providecommand{\definitionname}{Definition}
\providecommand{\lemmaname}{Lemma}
\providecommand{\propositionname}{Proposition}
\providecommand{\remarkname}{Remark}
\providecommand{\theoremname}{Theorem}

\title{Set Cover with Delay -- Clairvoyance is not Required}
\author{%
	\begin{tabular}{c}
		Yossi Azar\tabularnewline
		\texttt{azar@tau.ac.il}\tabularnewline
		{\small{}Tel Aviv University}\tabularnewline
		\tabularnewline
		Shay Kutten\tabularnewline
		\texttt{kutten@technion.ac.il}\tabularnewline
		{\small{}Technion - Israel Institute of Technology}\tabularnewline
	\end{tabular}\and%
	\begin{tabular}{c}
		Ashish Chiplunkar\tabularnewline
		\texttt{ashishc@iitd.ac.in}\tabularnewline
		{\small{}Indian Institute of Technology Delhi}\tabularnewline
		\tabularnewline
		Noam Touitou\tabularnewline
		\texttt{noamtouitou@mail.tau.ac.il}\tabularnewline
		{\small{}Tel Aviv University}\tabularnewline
\end{tabular}}

\begin{document}
\global\long\def\E{\mathbb{E}}%

\global\long\def\ADV{\operatorname{ADV}}%

\global\long\def\ALG{\operatorname{ALG}}%

\global\long\def\OPT{\operatorname{OPT}}%

\global\long\def\ON{\operatorname{ON}}%

\global\long\def\ONF{\operatorname{ONF}}%

\global\long\def\ONR{\operatorname{ONR}}%

\global\long\def\SCD{\operatorname{SCD}}%

\global\long\def\SCDALG{\operatorname{SCDALG}}%

\global\long\def\OSC{\operatorname{OSC}}%

\newcommand{\diff}[1]{\,\mathrm{d}#1}

\maketitle
\begin{abstract}
In most online problems with delay, clairvoyance (i.e. knowing the future delay of a request upon its arrival) is required for polylogarithmic competitiveness. In this paper, we show that this is not the case for set cover with delay (SCD) -- specifically, we present the first non-clairvoyant algorithm, which is $O(\log n \log m)$-competitive, where $n$ is the number of elements and $m$ is the number of sets. This matches the best known result for the classic online set cover (a special case of non-clairvoyant SCD). Moreover, clairvoyance does not allow for significant improvement -- we present lower bounds of $\Omega(\sqrt{\log n})$ and $\Omega(\sqrt{\log m})$ for SCD which apply for the clairvoyant case. 

In addition, the competitiveness of our algorithm does not depend on the number of requests. Such a guarantee on the size of the universe alone was not previously known even for the clairvoyant case -- the only previously-known algorithm (due to Carrasco et al.) is clairvoyant, with competitiveness that grows with the number of requests. 

For the special case of vertex cover with delay, we show a simpler, deterministic algorithm which is $3$-competitive (and also non-clairvoyant).
\end{abstract}
\thispagestyle{empty}
\setcounter{page}{0}
\clearpage

\section{Introduction}

In problems with delay, requests are released over a timeline. The algorithm must serve these requests by performing some action, which incurs a cost. While a request is pending (i.e. has been released but not yet served), the request accumulates delay cost. The goal of the algorithm is to minimize the sum of costs incurred in serving requests and the delay costs of requests.

There are two variants of such problems. In the clairvoyant variant, the delay function of a request (which determines the delay accumulation of that request over time) is revealed to the algorithm upon the release of the request. In the non-clairvoyant variant, at any point in time the algorithm is only aware of delay accumulated up to that point. 

Most online problems with delay do not admit competitive non-clairvoyant algorithms -- namely, there exist lower bounds for competitiveness which are polynomial in the size of the input space (e.g. the number of points in the metric space upon which requests are released). This is the case, for example, in the multilevel aggregation problem \cite{DBLP:conf/esa/BienkowskiBBCDF16,DBLP:conf/soda/BuchbinderFNT17}, the facility location problem \cite{DBLP:conf/focs/AzarT19} and the service with delay problem \cite{DBLP:conf/stoc/AzarGGP17}. However, these problems do admit \emph{clairvoyant} algorithms which are polylog-competitive. An additional such problem is that of matching with delay (presented in \cite{DBLP:conf/stoc/EmekKW16}), for which the only known algorithms are for when all requests have an identical, linear delay function (and are in particular clairvoyant). Rather surprisingly, we show in this paper that the online set cover with delay problem \emph{does} admit a competitive non-clairvoyant algorithm. 

In the online set cover with delay problem (SCD), a universe of elements and a family of sets are
known in advance. Requests then arrive over time on the
elements, and accumulate delay cost until served by the algorithm.
The algorithm may choose to buy a set at any time, at a cost
specific to that set (and known in advance to the algorithm). Buying
a set serves all pending requests (requests released but not yet
served) on elements of that set; future requests on those elements, that have yet to arrive when the set is bought, must be served separately at a future point in time. For that reason, a set may be bought an
unbounded number of times over the course of the algorithm. The goal
of an algorithm is to minimize the sum of the total buying cost and
the total delay cost. We note that one could also consider the problem in which sets are bought permanently, and cover future requests; however, it is easy to see that this problem is equivalent to the classic online set cover, and is thus of no additional interest. In Appendix \ref{sec:NonclairvoyantIsOptimal2}, we show that this problem is a special case of our problem.

As a variant of set cover, the SCD problem is very general, capturing many problems. Nevertheless, we give two possible motivations for the problem. 

\textbf{Summoning experts:} consider a company which
occasionally requires the help of experts. At any time, a problem
may arise which requires external assistance in some field, and
negatively impacts the performance of the company while unresolved.
At any time, the company may hire any one of a set of experts to
come to the company, solve all standing problems in that expert's
fields of expertise, and then leave. The company aims to minimize
the total cost of hiring experts, as well as the negative impact of
unresolved problems.

\textbf{Cluster-covering with delay:} suppose antennas generate data requests over time, which must be satisfied by an external server, with a cost to leaving a request pending. To satisfy an update request by an antenna, the server sends the data to a center antenna which transmits it at a certain radius, at a certain cost (which depends on the center antenna and the radius). All requests on antennas inside that radius are served by that transmission. This problem is a covering problem with delay costs, which can be described in terms of SCD. As an SCD instance, the elements are the antennas, and the sets are pairs of a center antenna and a (reasonable) transmission radius (the number of sets is quadratic in the number of antennas).

Carrasco \textit{et al.} \cite{DBLP:conf/latin/CarrascoPSV18} provided a clairvoyant algorithm for the SCD problem, which is $O(\log N)$ competitive (where $N$ is the number of the requests). However, as the number of requests becomes large, the competitive ratio of this algorithm tends to infinity -- even for a very small universe of elements and sets. Thus, this algorithm does not provide a guarantee in terms of the underlying input space, as we would like. In addition, their algorithm has exponential running time (through making oracle calls which compute optimal solutions for NP-hard problems).

In this paper, we present the first algorithm for SCD which is polylog-competitive in the size of the universe, which is also the first algorithm for the problem which runs in polynomial time. Surprisingly, this algorithm is also non-clairvoyant, showing that the SCD problem admits non-clairvoyant competitive algorithms. Our randomized algorithm is $O(\log n \log m)$-competitive, where $n$ is the number of elements and $m$ is the number of sets. In this paper, we show a reduction from the classic online set cover to SCD, which implies (due to \cite{KormanSetCoverRandomization}) that our upper bound is tight for a polynomial-time, non-clairvoyant algorithm for SCD. 

While our algorithm is optimal for the non-clairvoyant setting, one could wonder if there exists a \emph{clairvoyant} algorithm which performs significantly better -- especially considering the aforementioned problems, in which the gap between the clairvoyant and non-clairvoyant cases is huge. We answer this in the negative -- namely, we show lower bounds of $\Omega(\sqrt{\log n})$ and $\Omega(\sqrt {\log m} )$ on the competitiveness of any randomized clairvoyant algorithm, showing that there is no large gap which clairvoyant algorithms could bridge. Nevertheless, a quartic gap still exists, e.g. in the case that $m=
\Theta(n)$. We conjecture that the gap is in fact quadratic, and leave this as an open problem.

In this paper, we also consider the problem of vertex cover with delay (denoted VCD). In the VCD problem, vertices of graph are given, with a buying cost associated with each vertex. Requests on the edges of the graph arrive over time, and accumulate delay until served by buying a vertex touching the edge (at the cost of that vertex's price). This problem corresponds to SCD where every element is in exactly two sets.

\subsection{Our Results}

We denote as before the number of elements in an SCD instance by $n$, and the number of sets by $m$. We also define $k\le m$ to be the maximum number of sets to which a specific element may belong. We consider arbitrary (nondecreasing) continuous delay functions (not only linear functions). 

In this paper, we present:
\begin{enumerate}
	\item An $O(\log k\cdot\log n)$-competitive, randomized, non-clairvoyant algorithm for SCD, based on rounding of a newly-designed $O(\log k)$-competitive algorithm for the fractional version of SCD.  The competitive ratio of this algorithm is tight -- we show a reduction from (classic) online set cover to non-clairvoyant SCD.
	
	\item Lower bounds of $\Omega(\sqrt{\log k})$ and $\Omega(\sqrt{\log n})$ on competitiveness for \textbf{clairvoyant} SCD, showing that clairvoyance cannot improve competitiveness beyond a quadratic factor. 
	
	\item A simple, deterministic, non-clairvoyant algorithm for vertex cover with delay (VCD) which is $3$-competitive.
\end{enumerate}

Our randomized algorithm for SCD is the first (sub-polynomial competitive) non-clairvoyant algorithm for this problem. Moreover, this is the first algorithm which is polylog-competitive in the size of the universe (even among clairvoyant algorithms).

%

In the process of obtaining our $\Omega(\sqrt{\log k})$ and $\Omega(\sqrt{\log n})$
lower bounds, we in fact obtain an $\Omega(\sqrt{\log m})$ lower
bound (which immediately implies $\Omega(\sqrt{\log k})$ since $k\le m$).
The lower bounds also apply for the unweighted setting. These lower bounds improve over the lower bound of $\Omega(\log \log n)$ given in \cite{DBLP:conf/latin/CarrascoPSV18}\footnote{The lower bound of \cite{DBLP:conf/latin/CarrascoPSV18} shows $\Omega(\log N)$-competitiveness, but relies on a universe which is exponentially larger than the number of requests. As they mention in their paper, this therefore translates to an $\Omega(\log \log n)$ lower bound on competitiveness.}.

For VCD, while our algorithm is $3$-competitive, note that there is a lower bound of $2$. The lower bound uses a graph with
a single edge which is requested multiple times; this graph corresponds to the TCP acknowledgment problem,
analyzed in \cite{DBLP:conf/stoc/DoolyGS98}.

\begin{rem}
	While our $O(\log k \cdot \log n)$-competitive algorithm is presented for the case in which the sets and elements are known in advance, it can easily be modified for the case in which each element, as well as which of the sets contain it, becomes known to the algorithm only after the arrival of a request on that element. Moreover, the algorithm can in fact operate in the original setting of Carrasco \textit{et al.} \cite{DBLP:conf/latin/CarrascoPSV18}, as it does not need to know the family of sets itself, but rather the family of \emph{restrictions} of the sets to the elements that have already arrived. This can be done through standard doubling techniques applied to $\log n$ and $\log k$ (i.e. squaring of $n$ and $k$).
\end{rem}

\subsection{Our Techniques}

In the course of designing a non-clairvoyant algorithm for the SCD problem,
we also consider a fractional version of SCD. In this version, an
algorithm may choose to buy a fraction of a set at any moment.
Buying a fraction of a set partially serves requests present on an
element of that set, which causes them to accumulate less future
delay. As with the original version, a request is only served by
fractions bought after its arrival. Hence, the sum of fractions
bought for a single set over time is unbounded (i.e. a set may be bought many times).


In the \textbf{fractional $O(\log k)$-competitive algorithm}, each
request that can be served by a set contributes some amount to the
buying of that set. This amount depends exponentially on the delay
accumulated by that request, as well as the delay of previous requests.
Typically in algorithms with exponential contributions, these contributions
are summed. Interestingly, our algorithm instead chooses the maximum
of the contributions of the requests as the buying function of the
set. The choice of maximum over sum is crucial to the proof (using
sum instead of maximum would lead to a linear competitive ratio).

The analysis of this algorithm is based on dual fitting: we first present
a linear programming representation of the fractional SCD problem,
then use a feasible solution to the dual problem to charge the delay
of the algorithm to the optimum. This is the reason for using the
maximum in the buying function; each quantity satisfies a different
constraint in the dual, and choosing the maximum satisfies all constraints.
We then charge the buying cost of the algorithm to $O(\log k)$ times
its delay, which concludes the analysis.

Next, we design a randomized competitive algorithm for the integer version of SCD using \emph{2-level} randomized rounding of the fractional algorithm. At the \emph{top level}, we construct a \textbf{randomized $O(\log k\cdot\log N)$-competitive
algorithm for the integer version},
with $N$ the number of requests. The top-level rounding consists of maintaining
for each set a random threshold, and buying the set when the total
buying of that set in the \emph{fractional }algorithm exceeds the
threshold. In addition, special service of a request is performed
in the probabilistically unlikely event that the request is half-served in the fractional algorithm but is still pending in the rounding. Since in our
problem we may buy a set an unbounded number of times, we require
use of multiple subsequent thresholds. To analyze this, we make use
of Wald's equation for stopping time.

We add the \emph{bottom level} to improve the $O(\log k\cdot\log N)$-competitive algorithm to a
\textbf{randomized $O(\log k\cdot\log n)$-competitive algorithm for
the integer version}. The bottom level partitions time into phases for each element separately, and aggregates requests on that element that are released in the same phase. The competitive ratio of the resulting algorithm is asymptotically optimal for solving non-clairvoyant SCD in polynomial time, as shown by the reduction from the classic online set cover to non-clairvoyant SCD given in Appendix \ref{sec:NonclairvoyantIsOptimal2}.

Perhaps the most novel techniques in this paper are used for \textbf{the
lower bounds of $\Omega(\sqrt{\log k})$ and $\Omega(\sqrt{\log n})$} for the clairvoyant case. The lower bounds are obtained by a recursive construction. Given a recursive
instance for which any algorithm has a lower bound on the competitive
ratio, we amplify that bound by duplicating every set in the recursive
instance into two sets, one slightly more expensive than the other.
Both sets perform the same function with respect to the recursive instance,
but the algorithm also has an incentive to choose the expensive family
of sets, since they serve some additional requests. If the algorithm
chooses to buy a lot of expensive sets, the optimum releases another
copy of the recursive instance, serviceable only by expensive sets.
This forces the algorithm to buy the expensive sets twice; the optimum
only buys them once. If, on the other hand, the algorithm chooses
the inexpensive sets, it misses the opportunity to serve the additional
requests and the recursive instance simultaneously, and must serve
them separately.

The recursive description of our construction for the lower bounds
is significantly more natural than its iterative description. Few
lower bounds in online algorithms have this property -- another such
lower bound is found in \cite{DBLP:conf/soda/BansalC09}.

\textbf{The $3$-competitive deterministic algorithm for VCD} is simple and based on counters. This algorithm is only $k+1$ competitive for
general SCD, and is thus significantly worse than the previous randomized
algorithm that we have shown for general SCD.

\subsection{Other Related Work}

A different problem called online set cover is considered in \cite{DBLP:conf/stoc/AwerbuchAFL96},
in which the algorithm accumulates value for every element that arrives
on a bought set, and aims to maximize total value. This problem appears
to be fundamentally different from the online set cover in which we
minimize cost, in both techniques and results.


The problem of set cover in the online setting has seen much additional
work, e.g. in \cite{DBLP:journals/siamcomp/GrandoniGLMSS13,DBLP:conf/approx/BhawalkarGP14,DBLP:journals/tcs/DobrevEKKKKM17,DBLP:conf/infocom/PananjadyBV15,DBLP:conf/cocoa/AbshoffMH14}. The set cover problem has also been studied in the streaming model (e.g. \cite{DBLP:conf/icalp/EmekR14,DBLP:conf/soda/ChakrabartiW16}),  stochastic model (e.g. \cite{DBLP:journals/talg/ImNZ16}), dynamic model (e.g. \cite{DBLP:conf/stoc/GuptaK0P17}), and in the context of universal algorithms (e.g. \cite{DBLP:conf/stoc/JiaLNRS05,DBLP:journals/siamcomp/GrandoniGLMSS13}) and communication complexity  (e.g. \cite{DBLP:conf/icalp/Nisan02}).

There are known inapproximability results for the (offline) set cover
and vertex cover problems. In \cite{DBLP:conf/stoc/DinurS14} it
is shown that the offline set cover problem is unlikely to be approximable
in polynomial time to within a factor better than $\ln n$. For the
offline vertex cover, it is shown in \cite{DBLP:journals/jcss/KhotR08}
that it is NP hard to approximate within a factor better than $2$,
assuming the \emph{Unique Games Conjecture}. These results apply to
our SCD and VCD problems, as an instance of offline set cover (or
vertex cover) can be released at time $0$. Of course, these inapproximability
results do not constitute lower bounds for the online model, in which
unbounded computation is allowed -- unlike the information-theoretic lower bound of $\Omega(\sqrt{\log n})$ for SCD which is given in this paper.

The field of online problems with delay over time has been of interest
recently. This includes the problems of min-cost perfect matching with delays \cite{DBLP:conf/stoc/EmekKW16,DBLP:conf/soda/AzarCK17,DBLP:conf/approx/AshlagiACCGKMWW17,DBLP:conf/waoa/BienkowskiKS17,DBLP:journals/corr/abs-1804-08097,2018arXiv180603708A}, online service with delay  \cite{DBLP:conf/stoc/AzarGGP17,DBLP:conf/sirocco/BienkowskiKS18,DBLP:conf/focs/AzarT19} and multilevel aggregation \cite{DBLP:conf/esa/BienkowskiBBCDF16,DBLP:conf/soda/BuchbinderFNT17,DBLP:conf/focs/AzarT19}.


\subsection*{Paper Organization}
In Section \ref{sec:FractionalAlgorithm}, we present and analyze a fractional non-clairvoyant algorithm for SCD. In Section \ref{sec:RoundingAlgo}, we show how to round the previous algorithm in a non-clairvoyant manner to obtain our algorithm for the original (integral) SCD.  In Section \ref{sec:LowerBounds}, we show lower bounds for clairvoyant SCD. In Appendix \ref{sec:NonclairvoyantIsOptimal2}, we show that the algorithm obtained in Section \ref{sec:RoundingAlgo} is optimal for the non-clairvoyant case. In Section \ref{sec:VertexCover}, we give a simple, deterministic, non-clairvoyant algorithm for vertex cover with delay.
\section{\label{sec:Preliminaries}Preliminaries}

We denote the sets by $\{S_{i}\}_{i=1}^{m}$, with $m$ the number
of sets. We denote by $n$ the number of elements. We define $k$
to be the minimal number for which every element belongs to at most
$k$ sets. Requests $q_{j}$ arrive on the elements. We denote the
arrival time of request $q_{j}$ by $r_{j}$, and write (with a slight
abuse of notation) $q_{j}\in S_{i}$ if the element on which $q_{j}$
has been released belongs to the set $S_{i}$.

Each request $q_{j}$ has an arbitrary momentary
delay function $d_{j}(t)$, defined for $t\ge r_{j}$. The accumulated
delay of the request at time $t\ge r_{j}$ is defined to be $\int_{r_{j}}^{t}d_{j}(t)\diff{t}$.
At any time in which a request is pending, its momentary delay is
added to the cost of the algorithm; that is, the algorithm incurs
a cost of $\int_{r_{j}}^{\tau_{j}}d_{j}(t)\diff{t}$ (the accumulated delay
of $q_{j}$ at time $\tau_{j}$) for every request $q_{j}$, where $\tau_{j}$
is the time in which $q_{j}$ is served. Each set $S_{i}$ has a price
$c(S_{i})\ge1$ which the algorithm must pay when it decides to buy
the set. Buying a set serves all pending requests which belong to
the set (but does not affect future requests). The \emph{buying cost
}of an online algorithm $\ON$ is $\text{Cost}_{\ON}^{p}=\sum_{i}n_{i}\cdot c(S_{i})$,
where $n_{i}$ is the number of times $S_{i}$ has been bought by
the algorithm. The \emph{delay cost }of $\ON$ is $\text{Cost}_{\ON}^{d}=\sum_{j}\int_{r_{j}}^{\tau_{j}}d_{j}(t)\diff{t}$,
where $\tau_{j}$ is the time in which $q_{j}$ is served by the algorithm ($\tau_{j}$ is $\infty$ if $q_j$ is never  served by the algorithm)\footnote{We solve the more general problem in which the algorithm doesn't have to serve all requests (observe that the adversary can still force the algorithm to serve all requests by adding infinite delay at time infinity). This allows the problem to capture additional problems (e.g. prize-collecting problems, in which a penalty could be paid to avoid serving a specific request)}.

Overall, the cost of $\ON$ for the problem is $\text{Cost}_{\ON}=\text{Cost}_{\ON}^{p}+\text{Cost}_{\ON}^{d}$

\section{\label{sec:FractionalAlgorithm}The Non-Clairvoyant Algorithm for Fractional SCD}
We first describe a fractional relaxation of the (integer) set cover with delay problem. In this fractional relaxation, a set can be bought in parts. A fractional algorithm
determines for each set $S_{i}$ a nonnegative momentary buying function
$x_{i}(t)$. The total buying cost a fractional online algorithm $F$ incurs
is $\text{Cost}_{F}^{p}=\sum_{i}c(S_{i})\cdot\int_{0}^{\infty}x_{i}(t)\diff{t}$.

In the fractional version, a request can be partially served. Under
a fractional algorithm $F$, for any request $q_{j}$, and any set
$S_{i}$ such that $q_{j}\in S_{i}$, the set $S_{i}$ covers $q_{j}$
at a time $t\ge r_{j}$ by the amount $\int_{r_{j}}^{t}x_{i}(t^{\prime})\diff{t^{\prime}}$ (which is obviously nondecreasing as a function of $t$).
The total amount by which $q_{j}$ is covered at time $t$ is
\[
\gamma_{j}(t)=\sum_{i|q_{j}\in S_{i}}\int_{r_{j}}^{t}x_{i}(t^{\prime})\diff{t^{\prime}}
\]
If at time $t$ we have $\gamma_{j}(t)\ge1$, then $q_{j}$ is considered
served, and the algorithm does not incur delay. However, if $\gamma_{j}(t)<1$,
the algorithm $F$ incurs delay proportional to the uncovered fraction
of $q_{j}$. Formally, at time $t$ the request $q_{j}$ incurs $d_{j}^{F}(t)$
delay in $F$, where
\begin{equation}
d_{j}^{F}(t)=\begin{cases}
d_{j}(t)\cdot(1-\gamma_{j}(t)) & \text{if }\gamma_{j}(t)<1\\
0 & \text{otherwise}
\end{cases}\label{eq:d_j_ONF_definition}
\end{equation}

The delay cost of the algorithm is $\text{Cost}_{F}^{d}=\sum_{j}\int_{r_{j}}^{\infty}d_{j}^{F}(t)\diff{t}$.
The total cost of the fractional algorithm is thus
$\text{Cost}_{F}=\text{Cost}_{F}^{p}+\text{Cost}_{F}^{d}$.

\textbf{Description of the algorithm.} We now describe an online algorithm called $\ONF$ for the fractional
problem.

We define a total order $\preceq$ on requests, such that for any
two requests $q_{j_{1}},q_{j_{2}}$ if $r_{j_{1}}<r_{j_{2}}$ we have
$q_{j_{1}}\prec q_{j_{2}}$ (we break ties arbitrarily between requests
with equal arrival time).

At any time $t$, the algorithm does the following.

\noindent\fbox{\begin{minipage}[t]{1\columnwidth - 2\fboxsep - 2\fboxrule}%
\begin{enumerate}
\item For every request $q_{j}$, evaluate $d_{j}^{\ONF}(t)$ by its definition
in Equation \ref{eq:d_j_ONF_definition}.
\item For every set $S_{i}$ and request $q_{j}\in S_{i}$, define \[
D_{i}^{j}(t)=\sum_{j^{\prime}|q_{j^{\prime}}\in S_{i}\wedge q_{j^{\prime}}\preceq q_{j}}d_{j^{\prime}}^{\ONF}(t)\]
\item For every set $S_{i}$ and request $q_{j}\in S_{i}$, define
\[
x_{i}^{j}(t)=\frac{1}{k}\cdot\left(\frac{\ln(1+k)}{c(S_{i})}\cdot D_{i}^{j}(t)\right)\cdot e^{\frac{\ln(1+k)}{c(S_{i})}\int_{r_{j}}^{t}D_{i}^{j}(t^{\prime})\diff{t^{\prime}} }
\]
\item Buy every set $S_{i}$ according to $x_{i}(t)$, such that
\[
x_{i}(t)=\max_{j}x_{i}^{j}(t)
\]
\end{enumerate}
\end{minipage}}

This completes the description of the algorithm.

The intuition for the algorithm is that at any time $t$, the amount $\int_{r_j}^{t} D_i^j (t^{\prime}) \diff{t^{\prime}}$ is delay incurred by the algorithm until time $t$ that the optimum possibly avoided by buying $S_i$ at time $r_j$, and thus the algorithm wishes to minimize this amount. Thus, the request $q_j$ places some "demand" on the algorithm to buy $S_i$. Since this is true for any $q_j \in S_i$, the algorithm chooses the maximum of the demands as the buying function of $S_i$.

This demand $x_i^j (t)$ placed on the algorithm by $q_j$ to buy $S_i$ is related to $\int_{r_j}^{t} D_i^j (t^{\prime})\diff{t^{\prime}}$. If we wanted to make the total buying proportional to $\int_{r_j}^{t} D_i^j (t^{\prime}) \diff{t^{\prime}}$, it would sound reasonable to set $x_i^j (t)$ to be its derivative, namely $D_i^j (t)$. However, this would only be $\Omega(k)$-competitive, as demonstrated in Figure \ref{fig:LinearBuyingBad}. We thus want the total buying to be proportional to an expression exponential in  $\int_{r_j}^{t} D_i^j (t^{\prime}) \diff{t^{\prime}}$, which underlies the definition of $x_i^j(t)$ in our algorithm.

Denoting $X_{i}^{j}(t)=\int_{r_{j}}^{t}x_{i}^{j}(t^{\prime})\diff{t^{\prime}}$,
note that
\begin{equation}
X_{i}^{j}(t)=\frac{1}{k}\cdot\left[e^{\frac{\ln(1+k)}{c(S_{i})}\int_{r_{j}}^{t}D_{i}^{j}(t^{\prime})\diff{t^{\prime}}}-1\right]\label{eq:Integral}
\end{equation}

In the rest of this section, we prove the following theorem.
\begin{thm}
\label{thm:FractionalCompetitiveness}The algorithm for fractional
SCD described above is $O(\log k)$-competitive.
\end{thm}

We now analyze the algorithm for fractional SCD and prove Theorem \ref{thm:FractionalCompetitiveness}.

\begin{figure}
	\begin{centering}
		\caption{Linear Buying $\Omega(k)$ Example}
		\label{fig:LinearBuyingBad}
		\par\end{centering}
	\begin{raggedright}
		In this figure, there are $k-1$ elements, where each element is contained in $k$ sets of cost 1, one central set (which contains all elements) and $k-1$ peripheral sets (each contains exactly one element). Consider $k-1$ requests, one on each element, all arriving at time $0$. Their delay functions are identical, and go to infinity as time progresses.
		
		Consider an algorithm which buys sets linearly to the delay - that is, $x_i (t)= \max_j D_i^j (t)= \sum_{j|q_j\in S_i} d_j^{\ONF} (t)$. The momentary delay of every request contributes equally to the buying functions of the $k$ containing sets. Thus, the total fraction bought of peripheral sets is exactly $k-1$ times the total fraction bought of the central set. Consider the point in time in which all requests are half-covered (through symmetry, this happens for all requests at the same time, and must happen since the requests gather infinite delay). We have that the central set was bought for a fraction of exactly $\frac{1}{4}$ (which can again be seen through symmetry of the requests and their delay). Thus, the peripheral sets were bought for a fraction of $\frac{k-1}{4}$, for a total of $\frac{k}{4}$. Consider that the optimal solution costs $1$, as the optimum buys the central set at time $0$.
		~\\
		\par\end{raggedright}
	\centering{}\includegraphics{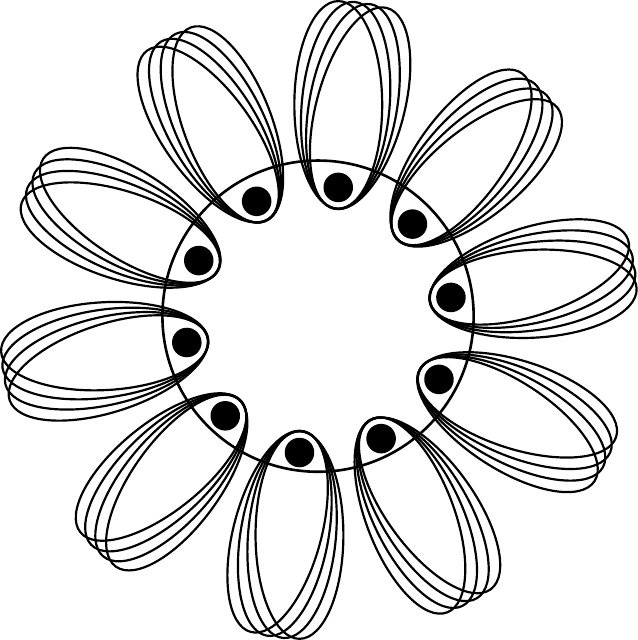}
\end{figure}

\subsection{\label{subsec:BoundingFractionalBuyingCost}Charging Buying Cost
	to Delay}

In this subsection we prove the following lemma.
\begin{lem}
	\label{lem:BuyingChargedToDelay}$\text{Cost}_{\ONF}^{p}\le2\ln(1+k)\cdot\text{Cost}_{\ONF}^{d}$
\end{lem}

\begin{proof}
	The proof is by charging the momentary buying cost at any time $t$
	to the $2\ln(1+k)$ times the momentary delay incurred by $\ONF$ at
	$t$. Let $q_{j}$ be some request released by time $t$. For every
	$i$ such that $q_{j}\in S_{i}$, we charge some amount $z_{i}^{j}(t)$
	to $d_{j}^{\ONF}(t)$. Denote by $j_{i}$ the request in $S_{i}$ such
	that
	\[
	x_{i}(t)=x_{i}^{j_{i}}(t)
	\]
	
	If $q_{j}\preceq q_{j_{i}}$, we choose
	\[
	z_{i}^{j}(t)=\frac{\ln(1+k)}{k}\cdot d_{j}^{\ONF}(t)\cdot e^{\frac{\ln(1+k)}{c(S_{i})}\int_{r_{j_{i}}}^{t}D_{i}^{j_{i}}(t^{\prime})\diff{t^{\prime}} }
	\]
	Otherwise, we choose $z_{i}^{j}(t)=0$. Note that for every set $S_{i}$
	we have $\sum_{j|q_{j}\in S_{i}}z_{i}^{j}(t)=c(S_{i})\cdot x_{i}(t)$,
	and thus the entire buying cost is charged.
	
	The total buying cost charged to a request $q_{j}$ at time $t$ is
	$\sum_{i|q_{j}\in S_{i}}z_{i}^{j}(t)$. We show that for any $j$
	we have
	\[
	\sum_{i|q_{j}\in S_{i}}z_{i}^{j}(t)\le2\ln(1+k)\cdot d_{j}^{\ONF}(t)
	\]
	
	Summing the previous equation over requests $q_{j}$ and integrating
	over time yields the lemma.
	
	If $d_{j}^{\ONF}(t)=0$ we have $z_{i}^{j}(t)=0$ for every $i$, as
	required. From now on, we assume that $d_{j}^{\ONF}(t)>0$.
	
	Denote by $T_{j}=\{i|q_{j}\in S_{i}\text{ and }z_{i}^{j}>0\}$. We
	have
	\begin{align*}
		\sum_{i|q_{j}\in S_{i}}z_{i}^{j}(t) & =\sum_{i\in T_{j}}z_{i}^{j}(t)\\
		& =\ln(1+k)\cdot d_{j}^{\ONF}(t)\cdot\sum_{i\in T_{j}}\frac{1}{k}\cdot e^{\frac{\ln(1+k)}{c(S_{i})}\int_{r_{j_{i}}}^{t}D_{i}^{j_{i}}(t^{\prime})\diff{t^{\prime}}}
	\end{align*}
	
	Now note that
	\begin{align*}
		\frac{1}{k}\cdot e^{\frac{\ln(1+k)}{c(S_{i})}\int_{r_{j_{i}}}^{t}D_{i}^{j_{i}}(t^{\prime})\diff{t^{\prime}}} & =\frac{1}{k}+X_{i}^{j_{i}}(t)\\
		& \le\frac{1}{k}+\int_{r_{j_{i}}}^{t}x_{i}(t^{\prime})\diff{t^{\prime}}\\
		& \le\frac{1}{k}+\int_{r_{j}}^{t}x_{i}(t^{\prime})\diff{t^{\prime}}
	\end{align*}
	where the equality is due to equation \ref{eq:Integral}, the first
	inequality is due to the definition of $X_{i}^{j_{i}}(t)$ and since
	$x_{i}(t)\ge x_{i}^{j_{i}}(t)$, and the last inequality is due to
	$q_{j}\preceq q_{j_{i}}$.
	
	Thus
	\[
	\sum_{i|q_{j}\in S_{i}}z_{i}^{j}(t)\le\ln(1+k)\cdot d_{j}^{\ONF}(t)\cdot\sum_{i\in T_{j}}\left(\frac{1}{k}+\int_{r_{j}}^{t}x_{i}(t^{\prime})\diff{t^{\prime}}\right)\le2\ln(1+k)\cdot d_{j}^{\ONF}(t)
	\]
	
	where the last inequality follows from $|T_{j}|\le k$, and from $\sum_{i|q_{j}\in S_{i}}\int_{r_{j}}^{t}x_{i}(t^{\prime})\diff{t^{\prime}}\le 1$
	(due to the assumption that $d_{j}^{\ONF}(t)>0$).
\end{proof}

\subsection{Charging Delay to Optimum}

In this subsection, we charge the delay of the algorithm to the optimum
via dual fitting.

\subsubsection{Linear Programming Formulation}

We formulate a linear programming instance for the fractional problem,
and observe its dual instance.

\textbf{Primal}

In the primal instance, the variables are:
\begin{itemize}
	\item $x_{i}(t)$ for a set $S_{i}$ and time $t$, which is the fraction
	by which the algorithm buys $S_{i}$ at time $t$.
	\item $p_{j}(t)$ for a request $q_{j}$ and time $t \ge r_j$, which is the fraction
	of $q_{j}$ not covered by bought sets at time $t$.
\end{itemize}
The LP instance is therefore:

Minimize:
\[
\sum_{i}\int_{0}^{\infty}c(S_{i})\cdot x_{i}(t)\diff{t}+\sum_{j}\int_{r_{j}}^{\infty}p_{j}(t)\cdot d_{j}(t)\diff{t}
\]
under the constraints:
\[
\forall j,t\ge r_j:\,p_{j}(t)+\sum_{i|q_{j}\in S_{i}}\int_{r_{j}}^{t}x_{i}(t^{\prime})\diff{t^{\prime}}\ge1
\]
\begin{align*}
	p_{j}(t)\ge0\,,\,x_{i}(t)\ge0
\end{align*}

\textbf{Dual}

Maximize:
\[
\sum_{j}\int_{r_j}^{\infty}y_{j}(t)\diff{t}
\]
under the constraints:
\begin{equation}
	\forall i,t:\,\sum_{j|q_{j}\in S_{i}\wedge r_{j}\le t}\int_{t}^{\infty}y_{j}(t^{\prime})\diff{t^{\prime}}\le c(S_{i})\text{\tag{C1}}\label{eq:C1}
\end{equation}
\begin{equation}
	\forall j,t\ge r_j:\,y_{j}(t)\le d_{j}(t)\tag{C2}\label{eq:C2}
\end{equation}
\[
y_{j}(t)\ge0
\]

\begin{rem}
	As we chose to consider time as continuous, the linear program described here has an infinite number of variables and constraints. This is merely a choice of presentation, as discretizing time would yield a standard, finite LP. Nevertheless, weak duality for this infinite LP (the only duality property used in this paper) holds (see e.g. \cite{ReilandWeakDuality}).
\end{rem}

\subsubsection{Charging Delay to Optimum via Dual Fitting }

We now charge the delay of the fractional algorithm to the cost of
the optimum.
\begin{lem}
	\label{lem:DelayChargedToOPT}$\text{Cost}_{\ONF}^{d}\le\text{Cost}_{OPT}$
\end{lem}

\begin{proof}
	The proof is by finding a solution to the dual problem, such that
	the goal function value of the solution is equal to the delay of the
	algorithm.
	
	For every request $q_{j}$ and time $t$, we assign $y_{j}(t)=d_{j}^{\ONF}(t)$.
	This assignment satisfies that the goal function is the total delay
	incurred by the algorithm.
	
	Note that the \ref{eq:C2} constraints trivially hold, since $d_{j}^{\ONF}(t)\le d_{j}(t)$
	for any $j,t$. Now observe the \ref{eq:C1} constraints. For any
	time $t$ and a set $S_{i}$, the resulting \ref{eq:C1} constraint
	is implied by the \ref{eq:C1} constraint of time $r_{j}$ and the
	set $S_{i}$, with $q_{j}$ being the last request released by time
	$t$. We thus restrict ourselves to \ref{eq:C1} constraints of time
	$r_{j}$ for some $j$.
	
	For a request $q_{j}$ and a set $S_{i}$, we need to show:
	\[
	\sum_{j^{\prime}|q_{j^{\prime}}\in S_{i}\wedge q_{j^{\prime}}\preceq q_{j}}\int_{r_{j}}^{\infty}d_{j^{\prime}}^{\ONF}(t^{\prime})\diff{t^{\prime}}\le c(S_{i})
	\]
	Using the definition of $D_{i}^{j}(t)$, we need to show:
	\[
	\int_{r_{j}}^{\infty}D_{i}^{j}(t)\diff{t}\le c(S_{i})
	\]
	Define $t_{0}$ to be the minimal time (possibly $\infty$) such that
	for all $t\ge t_{0}$ we have $D_{i}^{j}(t)=0$. We must have that
	$\int_{r_{j}}^{t_{0}}x_{i}(t)\diff{t}\le1$; otherwise, all requests $q_{j^{\prime}}\in S_{i}$
	such that $q_{j^{\prime}}\preceq q_{j}$ will be completed before
	$t_{0}$, in contradiction to $t_{0}$'s minimality. Thus we have
	\begin{align*}
		1 & \ge\int_{r_{j}}^{t_{0}}x_{i}(t)\diff{t}\ge\int_{r_{j}}^{t_{0}}x_{i}^{j}(t)\diff{t}\\
		& =\frac{1}{k}\left[e^{\frac{\ln(1+k)}{c(S_{i})}\int_{r_{j}}^{t_{0}}D_{i}^{j}(t)\diff{t}}-1\right]
	\end{align*}
	where the second inequality is due to the definition of $x_{i}(t)$,
	and the equality is due to equation \ref{eq:Integral}. This yields
	\[
	(1+k)^{\frac{1}{c(S_{i})}\int_{r_{j}}^{t_{0}}D_{i}^{j}(t)\diff{t}}\le1+k
	\]
	and thus
	\[
	\int_{r_{j}}^{\infty}D_{i}^{j}(t)\diff{t}=\int_{r_{j}}^{t_{0}}D_{i}^{j}(t)\diff{t}\le c(S_{i})
	\]
	as required.
\end{proof}
We can now prove the main theorem.
\begin{proof}
	(of Theorem \ref{thm:FractionalCompetitiveness}) Using Lemmas \ref{lem:BuyingChargedToDelay}
	and \ref{lem:DelayChargedToOPT}, we have
	\begin{align*}
		\text{Cost}_{\ONF} & =\text{Cost}_{\ONF}^{p}+\text{Cost}_{\ONF}^{d}\\
		& \le(2\ln(1+k)+1)\cdot\text{Cost}_{\ONF}^{d}\\
		& \le(2\ln(1+k)+1)\cdot\text{Cost}_{OPT}
	\end{align*}
	
	as required.
\end{proof}
\begin{rem}
	For the more difficult delay model in which a partially served request
	$q_{j}$ incurs delay $d_{j}^{\ONF}(t)=d_{j}(t)$ instead of $d_{j}^{\ONF}(t)=d_{j}(t)\cdot(1-\gamma_{j}(t))$
	in $\ONF$, this algorithm is still $O(\log k)$ competitive against
	the fractional optimum in the easier delay model. The proof is identical.
\end{rem}

\section{\label{sec:RoundingAlgo}Randomized Algorithm for SCD by Rounding}

In this section, we describe a non-clairvoyant, polynomial-time randomized algorithm which is $O(\log k\cdot\log n)$-competitive
for integral SCD. Our randomized algorithm uses randomized rounding
of the fractional algorithm of Section \ref{sec:FractionalAlgorithm}.
We describe the rounding in two steps. First, we show a somewhat simpler
algorithm which is $O(\log k\cdot\log N)$-competitive. We then modify
this algorithm to obtain a $O(\log k\cdot\log n)$-competitive algorithm.

The rounding of the fractional algorithm of section \ref{sec:FractionalAlgorithm}
costs the randomized integral algorithm of this section a multiplicative
factor of $\log n$ over that fractional algorithm.

Denote by $x_{i}(t)$ the fractional buying function in the algorithm
$\ONF$ of Section \ref{sec:FractionalAlgorithm}. For a request $q_{j}$,
we denote by $S_{i_{j}}$ the least expensive set containing $q_{j}$;
that is, $i_{j}=\arg\min_{i|q_{j}\in S_{i}}c(S_{i})$.

For every request $q_{j}$, we denote the total covering of $q_{j}$
at time $t$ in $\ONF$ by $\gamma_{j}(t)$, where
\[
\gamma_{j}(t)=\sum_{i|q_{j}\in S_{i}}\int_{r_{j}}^{t}x_{i}(t^{\prime})\diff{t^{\prime}}
\]

We denote by $t_{j}$ the first time in which $\gamma_{j}(t)=\frac{1}{2}$.
\subsection*{\label{subsec:RequestRoundingAlgo}$O(\log k\cdot\log N)$-Competitive
Rounding}

We now describe a randomized integral algorithm, called $\ONR$, which is $O(\log k\cdot\log N)$
competitive with respect to the fractional optimum, with $N$ the
total number of requests. We assume a-priori knowledge of $N$ for
the algorithm.

The randomized integral algorithm runs the fractional algorithm of
Section \ref{sec:FractionalAlgorithm} in the background, and thus
has knowledge of the function $x_{i}(t)$ for every $i$. The algorithm
does the following.

\noindent\fbox{\begin{minipage}[t]{1\columnwidth - 2\fboxsep - 2\fboxrule}%
\begin{enumerate}
\item At time $0$:
\begin{enumerate}
\item For every set $S_{i}$, choose $\Lambda_{i}$ from the range $[0,\frac{1}{2\ln N}]$
uniformly and independently, and set $\tau_{i}=0$.
\end{enumerate}
\item At time $t$:
\begin{enumerate}
\item \label{enu:TypeA}For every $i$, if $\int_{\tau_{i}}^{t}x_{i}(t^{\prime})\diff{t^{\prime}}\ge\Lambda_{i}$
then:
\begin{enumerate}
\item Buy $S_{i}$.
\item Assign to $\Lambda_{i}$ a new value drawn independently and uniformly
from $[0,\frac{1}{2\ln N}]$.
\item Assign $\tau_{i}=t$.
\end{enumerate}
\item \label{enu:TypeB}If there exists a pending request $q_{j}$ such
that $t\ge t_{j}$, buy $S_{i_{j}}$.
\end{enumerate}
\end{enumerate}
\end{minipage}}

We refer to the buying of sets at Step \ref{enu:TypeA} as ``type
a'', and to the buying of sets at Step \ref{enu:TypeB} as ``type
b''.

The intuition for the randomized rounding scheme is that we would like the probability of buying a set in a certain interval of time to be proportional to the fraction of that set bought by the fractional algorithm in that interval, independently of the other sets. This is achieved by the "type a" buying. However, using "type a" alone is problematic. Consider, for example, a request on an element in $k$ sets, such that the fractional algorithm buys $\frac{1}{k}$ of each of the sets to cover the request. Since the probability of buying a set is independent of other sets, there exists a probability that the randomized algorithm would not buy any of the $k$ sets, leaving the request unserved. This bears possibly infinite delay cost for the rounding algorithm, which is not incurred by the underlying fractional algorithm.

The "type b" buying solves this problem, by serving a pending request \emph{deterministically} when it is covered in the underlying fractional algorithm, through buying the cheapest set containing that request. This special service for the request might be expensive, but its probability is low, yielding low expected cost. This is ensured by the $2\log N$ "speedup" given to the "type a" buying, through choosing the thresholds $\Lambda _i$ from $[0,\frac{1}{2\ln N}]$ (rather than $[0,1]$).

\begin{thm}
\label{thm:RequestRandomizedCompetitive}The randomized algorithm
for SCD described above is $O(\log k\cdot\log N)$-competitive.
\end{thm}

The proof of Theorem \ref{thm:RequestRandomizedCompetitive} is given
in Appendix \ref{subsecAppendix:RequestRoundingProof}.

\subsection*{\label{subsec:ElementRoundingAlgo}Improved $O(\log k\cdot\log n)$-Competitive
Rounding}

By modifying the $O(\log k\cdot\log N)$-competitive randomized rounding, we prove the following theorem.

\begin{thm}
	\label{thm:ElementRandomizedCompetitive}There exists a non-clairvoyant, randomized $O(\log k\cdot\log n)$-competitive algorithm for SCD.
\end{thm}

The modified rounding algorithm and its analysis appear in Appendix \ref{subsecAppendix:ElementRoundingProof}.

\section{\label{sec:LowerBounds}Lower Bounds for Clairvoyant SCD}

\global\long\def\S{H}%

\global\long\def\T{\mathcal{T}}%

In this section, we show $\Omega(\sqrt{\log k})$ and $\Omega(\sqrt{\log n})$
lower bounds on competitiveness for any randomized, clairvoyant algorithm
for SCD or fractional SCD. While the lower bounds use instances in
which different sets can have different costs, these instances can
be modified to obtain instances with identical set costs. This implies
that the lower bounds also apply to the unweighted setting. This modification
is shown in Subsection \ref{subsec:LowerBoundUnweighted}.

This section shows the following theorem.
\begin{thm}
\label{thm:LowerBoundTheorem}Any randomized algorithm for SCD or
fractional SCD is both $\Omega(\sqrt{\log k})$-competitive and $\Omega(\sqrt{\log n})$-competitive.
\end{thm}

In proving Theorem \ref{thm:LowerBoundTheorem}, we show a lower bound
on competitiveness of a deterministic fractional algorithm against
an integral optimum. Showing this is enough to prove the theorem,
since any randomized online algorithm (fractional or integral) can
be converted to a deterministic fractional online algorithm with identical (or lesser)
cost. This follows from setting the momentary buying function of a set at a given time to be the expectation of that value in the randomized algorithm. Since the
optimum is integral, the bound also holds for integral SCD, as the
theorem states. Therefore, we only consider deterministic fractional
online algorithms henceforth.

We show our lower bounds by constructing a set of SCD instances, $\left\{ I_{i}\right\} _{i=0}^{\infty}$.
For each $i\ge0$, the SCD instance $I_{i}$ contains $2^{i}$ sets
and $3^{i}$ elements. We show that any algorithm must be $\Omega(\sqrt{i})$-competitive
for $I_{i}$, which is both $\Omega(\sqrt{\log m})$ and $\Omega(\sqrt{\log n})$.
Noting that $k\le m$, we also have $\Omega(\sqrt{\log k})$ as required.

The instance $I_{i}$ exists within the time interval $[0,3^{i})$.
That is, no request of $I_{i}$ is released before time $0$, and
at time $3^{i}$ the optimum has served all requests in $I_{i}$,
and the algorithm has incurred a high enough cost.

We define the sequence $(c_{i})_{i=0}^{\infty}$, which is used in the construction of $I_i$. The sequence is defined recursively, such that $c_{0}=1$ and for any $i\ge1$, we have that
\[
c_{i}=c_{i-1}+\frac{1}{12c_{i-1}}
\]

We now describe the recursive construction of the instance $I_{i}$.
We first describe the universe of $I_{i}$, which consists of its
sets and elements. We then describe the requests of $I_{i}$.

\subsubsection*{Universe of $I_{i}$:}

For the base instance
$I_{0}$, the universe consists of a single element $e$ and a single
set $S=\{e\}$. We have that $c(S)=1$.

For $i\ge 1$, the recursive construction of $I_i$ using $I_{i-1}$ is as follows. Denote by $E_{i-1}$ the elements in the universe of
$I_{i-1}$, and by $\S_{i-1}$ the family of sets in the universe
of $I_{i-1}$. For the construction of $I_{i}$, consider three disjoint
copies of $E_{i-1}$ and $\S_{i-1}$. For $l\in\{1,2,3\}$, we denote
by $E_{i-1}^{l}$ and $\S_{i-1}^{l}$ the $l$'th copy of $E_{i-1}$
and $\S_{i-1}$, respectively. We denote by $S^{l}$ the copy of the
set $S\in\S_{i-1}$ in $\S_{i-1}^{l}$. Similarly, we denote by $e^{l}$
the copy of an element $e\in E_{i-1}$ in $E_{i-1}^{l}$.

The universe of $I_{i}$ consists of:
\begin{itemize}
\item The elements $E_{i}=E_{i-1}^{1}\cup E_{i-1}^{2}\cup E_{i-1}^{3}$.
\item The family of sets $\S_{i}=\T_{1}\cup\T_{2}$, where $\T_{1}$ and $\T_{2}$ are defined below.
\end{itemize}

We define:
\begin{itemize}
\item The family of sets $\T_{1}=\{S^{1}\cup S^{2}|S\in\S_{i-1}\}$. A set
$T\in\T_{1}$ formed from $S\in\S_{i-1}$ has cost $c(T)=c(S)$.
\item The family of sets $\T_{2}=\{S^{1}\cup S^{3}|S\in\S_{i-1}\}$. A set
$T\in\T_{2}$ formed from $S\in\S_{i-1}$ has cost $c(T)=(1+\alpha_{i})\cdot c(S)$,
with $\alpha_{i}=\frac{1}{2c_{i-1}}$.
\end{itemize}

\begin{figure}
\begin{centering}
\caption{The Universes of $I_{0}$, $I_{1}$ and $I_{2}$}
\par\end{centering}
\begin{raggedright}
This figure shows the universes of $I_{0}$, $I_{1}$ and $I_{2}$.
In the figure, each element is a point and the sets are the bodies
containing them, where each set has a distinct color. The costs of
the sets are also shown in the figure. The figure shows how three copies of the set of elements $E_{i-1}$ (of the  instance $I_{i-1}$) appear in $I_i$ -- the copy $E_{i-1}^1$ appears at the top of $I_i$'s visualization, the copy $E_{i-1}^2$ appears at the bottom-left, and the copy $E_{i-1}^3$ appears at the bottom-right.\\
~\\
\par\end{raggedright}
\centering{}\includegraphics[scale=0.9]{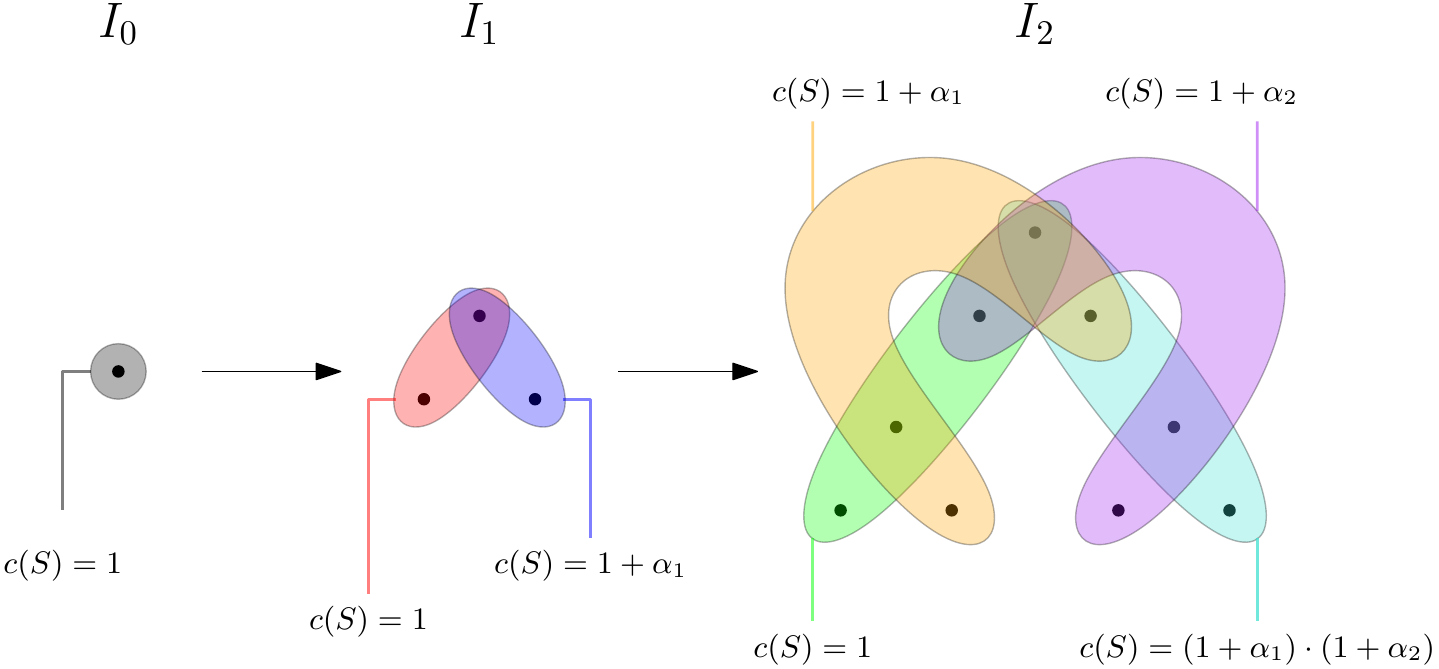}
\end{figure}
\subsubsection*{Requests of $I_{i}$:}

We first describe a type of request used in our construction. Let
$S$ be a set such that there exists an element $e\in S$ such that
$e$ is in no other set besides $S$ (we call $e$ unique to $S$).
For times $a,b$ such that $a<b$, we define a request $q_{a}^{b}(S)$
that can be released at any time $r\le a$ on an element unique
to $S$, and satisfies:
\begin{enumerate}
\item $\int_{r}^{a}d_{j}(t)\diff{t}=0$
\item $\int_{r}^{b}d_{j}(t)\diff{t}\ge c(S)$
\end{enumerate}
\begin{rem}
For the degenerate case of set cover with deadlines, when observing
a request with deadline at time $b$, it can be said to accumulate
$0$ delay until any time before $b$, and infinite delay until time
$b$. Therefore, deadline requests can function as $q_{a}^{b}(S)$
requests. Since all requests used in our construction are $q_{a}^{b}(S)$
requests for some $a,b,S$, our lower bound applies for set cover
with deadlines as well.
\end{rem}

To use those $q_{a}^{b}(S)$ requests, we require the following proposition,
which states that a $q_{a}^{b}(S)$ request can be released on every
$S$.
\begin{prop}
For every set $T\in\S_{i}$, there exists an element $e\in E_{i}$
unique to $T$.
\end{prop}

\begin{proof}
By induction on $i$. For the base case, this holds since there is
only a single set with a single element. Assuming the proposition
holds for $I_{i-1}$, we show that it holds for $I_{i}$ by observing
that there exists $S\in\S_{i-1}$ such that $T=S^{1}\cup S^{l}$ for
$l\in\{2,3\}$. Via induction, there exists an element $e\in E_{i-1}$
such that $e\in S$ and $e\notin S^{\prime}$ for every $S^{\prime}\in\S_{i-1}$
such that $S^{\prime}\neq S$. Choosing the element $e^{l}$ yields
the proposition.
\end{proof}

\textbf{\emph{Base case of $I_{0}$}}\emph{ -- }at time $0$, the
request $q_{0}^{1}(S)$ is released on the single element $e$.

\textbf{\emph{Recursive construction of $I_{i}$ using $I_{i-1}$
}}-- we define $C(I_{i})$ to be $\sum_{S\in\S_{i}}c(S)$. We now
define the instance $I_{i}$:

\noindent\fbox{\begin{minipage}[t]{1\columnwidth - 2\fboxsep - 2\fboxrule}%
\begin{enumerate}
\item At time $0$:
\begin{enumerate} 
\item \label{enu:Singletons}Release $q_{2\cdot3^{i-1}}^{3^{i}}(T)$ for
every $T\in\T_{2}$.
\item Release $I_{i-1}$ on the elements $E_{i-1}^{1}$ (see Remark (\nameref{rem:UnionInstance})).
\end{enumerate}
\item At time $3^{i-1}$:
\begin{enumerate} 
\item \label{enu:Case1}If the algorithm has bought sets of $\T_{2}$ at
a total cost of at least $\frac{1}{2}\cdot(1+\alpha_{i})\cdot C(I_{i-1})$,
release $(1+\alpha_{i})I_{i-1}$ on the elements $E_{i-1}^{3}$ (see
Remark (\nameref{rem:FancyInstance})).
\item \label{enu:Case2}Otherwise, release $I_{i-1}$ on the elements of
$E_{i-1}^{2}$ (see Remark (\nameref{rem:CheapInstance})).
\end{enumerate}
\end{enumerate}
\end{minipage}}

The construction of $I_{i}$ includes releasing copies of $I_{i-1}$
on the elements $E_{i-1}^{l}$, for $l\in\{1,2,3\}$. The following
remarks make this well-defined.
\begin{rem*}[a]
\label{rem:UnionInstance}\textbf{The $I_{i-1}$ on $E_{i-1}^{1}$}:
every set $S\in\S_{i-1}$ forms two sets in $\S_{i}$, which are $T_{1}=S^{1}\cup S^{2}\in\T_{1}$
and $T_{2}=S^{1}\cup S^{3}\in\T_{2}$. The $I_{i-1}$ construction
on $E_{i-1}^{1}$ treats buying either of these sets as buying the
set $S$. That is, it treats the sum of the momentary buying of $T_{1}$
and of $T_{2}$ as the momentary buying of $S$.
\end{rem*}
\begin{rem*}[b]
\label{rem:CheapInstance}\textbf{The $I_{i-1}$ on $E_{i-1}^{2}$}:
in this instance, for every set $S\in\S_{i-1}$, the $I_{i-1}$ construction
treats buying $T_{1}=S^{1}\cup S^{2}\in\T_{1}$ as buying $S$.
\end{rem*}
\begin{rem*}[c]
\label{rem:FancyInstance}\textbf{The scaled $(1+\alpha_{i})I_{i-1}$
on $E_{i-1}^{3}$}: similarly to Remark \ref{rem:CheapInstance},
in this instance, for every set $S\in\S_{i-1}$, the $I_{i-1}$ construction
treats buying $T_{2}=S^{1}\cup S^{3}\in\T_{2}$ as buying $S$. In
addition, since the sets of $\T_{2}$ are $(1+\alpha_{i})$-times
more expensive than the original sets of $\S_{i-1}$, the delays of
the jobs in $I_{i-1}$ are also scaled by $1+\alpha_{i}$ in order
to maintain the $I_{i-1}$ instance. We denote this scaled instance
by $(1+\alpha_{i})I_{i-1}$.
\end{rem*}

\subsection{\label{subsec:LowerBoundAnalysis} Analysis of Lower Bounds}

We show that any online fractional algorithm is at least $c_{i}$
competitive on $I_{i}$ with respect to the integral optimum.
\begin{lem}
	\label{lem:OptimumGood}The optimal integral algorithm can serve $I_{i}$
	by time $3^{i}$ with no delay cost by buying every set in $\S_{i}$
	exactly once, for a total cost of $C(I_{i})$.
\end{lem}

\begin{proof}
	Via induction on $i$. For the base case of $i=0$, the optimal algorithm
	buys the single set $S$ at time $0$ and pays $c(S)=C(I_{0})$. Now,
	for $i\ge1$, suppose the optimum can serve the instance $I_{i-1}$
	according to the lemma. We observe the optimum in $I_{i}$ according
	to the cases in the release of $I_{i}$:
	
	\textbf{Case \ref{enu:Case1}:}
	
	In this case, the optimum could have served $I_{i-1}$ on $E_{i-1}^{1}$
	by time $3^{i-1}$ by buying each set of $\T_{1}$ exactly once, with
	no delay cost. It could then serve $(1+\alpha_{i})I_{i-1}$ on $E_{i-1}^{3}$
	by time $2\cdot3^{i-1}$ by buying each set of $\T_{2}$ exactly once,
	with no delay cost. Since the optimum has bought all of $\T_{2}$,
	the requests released on step \ref{enu:Singletons} have also been
	served before incurring delay. The lemma thus holds for this case.
	
	\textbf{Case \ref{enu:Case2}:}
	
	In this case, the optimum could have served $I_{i-1}$ on $E_{i-1}^{1}$
	by time $3^{i-1}$ by buying each set of $\T_{2}$ exactly once, with
	no delay cost. It could then serve $I_{i-1}$ on $E_{i-1}^{2}$ by
	time $2\cdot3^{i-1}$ by buying each set of $\T_{1}$ exactly once,
	with no delay cost. Since the optimum has bought all of $\T_{2}$,
	the requests released on step \ref{enu:Singletons} have again been
	served before incurring delay. The lemma thus holds for this case
	as well.
\end{proof}
We now analyze the cost of the algorithm.
\begin{lem}
	\label{lem:AlgorithmBad}Any online algorithm has a cost of at least
	$c_{i}\cdot C(I_{i})$ on $I_{i}$ by time $3^{i}$.
\end{lem}

\begin{proof}
	By induction on $i$.
	
	For $i=0$, observe the algorithm at time $1$. Denoting by $\Gamma_{S}$
	the total buying of the single set $S$ by the algorithm by time $1$,
	the algorithm has a cost of at least
	\[
	c(S)\cdot\Gamma_{S}+(1-\Gamma_{S})\cdot\int_{0}^{1}d_{q_{0}^{1}(S)}(t)\diff{t}\ge c(S)=C(I_{0})
	\]
	where the inequality is due to the definition of $q_{0}^{1}(S)$.
	This finishes the base case of the induction.
	
	For the case that $i\ge1$, assume that the lemma holds for $i-1$.
	We show that it holds for $i$.
	
	Fix any algorithm for $I_{i}$. We denote by $\Gamma$ the total buying
	cost of the algorithm in the time interval $[0,3^{i-1})$ for sets
	of $\T_{2}$. We again split into cases according to the chosen branch
	in the construction of $I_{i}$.
	
	\textbf{Case \ref{enu:Case1}:}
	
	In this case we have $\Gamma\ge\frac{1}{2}\cdot(1+\alpha_{i})\cdot C(I_{i-1})$.
	From the definition of the first $I_{i-1}$ released, the adversary
	is oblivious to whether a copy of $S\in\S_{i-1}$ came from $\T_{1}$
	or $\T_{2}$. Using the induction hypothesis, any online algorithm
	for this instance incurs a cost of at least $c_{i-1}\cdot C(I_{i-1})$
	by time $3^{i-1}$, including the algorithm in which buying sets from
	$\T_{2}$ are replaced with buying the equivalent sets from $\T_{1}$.
	Such a modified online algorithm would cost $\frac{\alpha_{i}}{1+\alpha_{i}}\Gamma$
	less than the current algorithm, which is at least $\frac{\alpha_{i}}{2}\cdot C(I_{i-1})$.
	Therefore, the algorithm pays at least $(c_{i-1}+\frac{\alpha_{i}}{2})\cdot C(I_{i-1})$
	in the interval $[0,3^{i-1})$.
	
	As for the second instance $(1+\alpha_{i})I_{i-1}$, the algorithm
	must pay at least $(1+\alpha_{i})\cdot c_{i-1}\cdot C(I_{i-1})$ by
	time $2\cdot3^{i-1}$ via induction.
	
	Overall, the algorithm pays by time $3^{i}$ at least
	\begin{align*}
	\left(\left(c_{i-1}+\frac{\alpha_{i}}{2}\right)\cdot C(I_{i-1})\right)+ & \left((1+\alpha_{i})\cdot c_{i-1}\cdot C(I_{i-1})\right)\\
	& =\left((2+\alpha_{i})c_{i-1}+\frac{\alpha_{i}}{2}\right)\cdot C(I_{i-1})\\
	& =c_{i-1}\cdot C(I_{i})+\frac{\alpha_{i}}{2}\cdot C(I_{i-1})\\
	& \ge\left(c_{i-1}+\frac{\alpha_{i}}{6}\right)\cdot C(I_{i})\\
	& =\left(c_{i-1}+\frac{1}{12c_{i-1}}\right)\cdot C(I_{i})
	\end{align*}
	where the inequality is due to $C(I_{i})=(2+\alpha_{i})C(I_{i-1})\le3C(I_{i-1})$.
	
	\textbf{Case \ref{enu:Case2}:}
	
	In this case we have $\Gamma<\frac{1}{2}\cdot(1+\alpha_{i})\cdot C(I_{i-1})$.
	For the first $I_{i-1}$ instance, the algorithm pays at least $c_{i-1}\cdot C(I_{i-1})+\Gamma\cdot\frac{\alpha_{i}}{1+\alpha_{i}}$
	by time $3^{i-1}$, similar to the previous case.
	
	For the second $I_{i-1}$ instance, released on $E_{i-1}^{2}$, the
	algorithm must pay via induction at least $c_{i-1}\cdot C(I_{i-1})$
	by time $2\cdot3^{i-1}$. Since sets of $\T_{2}$ do not satisfy requests
	in this instance, this cost is either in buying sets of $\T_{1}$
	or in delay of requests from that $I_{i-1}$ instance.
	
	In addition to the two $I_{i-1}$ instances, due to the $q_{2\cdot3^{i-1}}^{3^{i}}(S)$
	requests released in step \ref{enu:Singletons}, the algorithm has
	a cost of at least $\left(\sum_{T\in\T_{2}}c(T)\right)-\Gamma=(1+\alpha_{i})C(I_{i-1})-\Gamma$
	during the interval $[1,3)$ in either buying sets of $\T_{2}$ in
	order to finish these requests, or in delay by those requests (using
	a similar argument to that in the base case). Overall, the algorithm
	has a cost of at least
	\begin{align*}
	\left(c_{i-1}\cdot C(I_{i-1})+\Gamma\cdot\frac{\alpha_{i}}{1+\alpha_{i}}\right)+ & \left(c_{i-1}\cdot C(I_{i-1})\right)+\left((1+\alpha_{i})C(I_{i-1})-\Gamma\right)\\
	& =(2c_{i-1}+1+\alpha_{i})\cdot C(I_{i-1})-\frac{1}{1+\alpha_{i}}\Gamma\\
	& \ge(2c_{i-1}+1+\alpha_{i})\cdot C(I_{i-1})-\frac{1}{2}C(I_{i-1})\\
	& =\left(2c_{i-1}+\frac{1}{2}+\alpha_{i}\right)\cdot C(I_{i-1})\\
	& =\left((2+\alpha_{i})c_{i-1}+\frac{1}{2}+(1-c_{i-1})\alpha_{i}\right)\cdot C(I_{i-1})\\
	& =c_{i-1}\cdot C(I_{i})+\left(\frac{1}{2}+\frac{1}{2c_{i-1}}-\frac{1}{2}\right)\cdot C(I_{i-1})\\
	& \ge\left(c_{i-1}+\frac{1}{6c_{i-1}}\right)\cdot C(I_{i})\ge c_{i}\cdot C(I_{i})
	\end{align*}
	where the fourth equality and the second inequality are due to $C(I_{i})=(2+\alpha_{i})C(I_{i-1})\le3C(I_{i-1})$,
	and the fourth equality uses the definition of $\alpha_{i}$.
\end{proof}
\begin{proof}
	(of Theorem \ref{thm:LowerBoundTheorem}) Lemmas \ref{lem:OptimumGood}
	and \ref{lem:AlgorithmBad} immediately imply that any deterministic
	fractional algorithm is at least $c_{i}$-competitive on $I_{i}$
	with respect to the integral optimum. Solving the recurrence in the
	definition of $c_{i}$, we have that $c_{i}=\Omega(\sqrt{i})$. To
	observe this, note that for every $i$, the first index $i^{\prime}\ge i$
	such that $c_{i^{\prime}}\ge c_{i}+1$ is at most $O(c_{i})$ larger
	than $i$. Since $k\le m=2^{i}$ and $n=3^{i}$, this provides lower
	bounds of $\Omega(\sqrt{\log k})$ and $\Omega(\sqrt{\log n})$ for
	deterministic algorithms for fractional SCD. As stated before, this
	implies the same lower bound for randomized algorithms for both SCD
	and fractional SCD.
\end{proof}

\subsection{\label{subsec:LowerBoundUnweighted} Reduction to Unweighted}

The lower bound above uses weighted instances, in which sets may have
different costs. In this subsection, we describe how to convert a
weighted instance to an unweighted instance, in which all set costs
are equal. This conversion maintains both the $\Omega(\sqrt{\log k})$
and $\Omega(\sqrt{\log n})$ lower bounds on competitiveness. The
conversion consists of creating multiple copies of each element, and
converting each original set to multiple sets of cost $1$. The cost
of the original set affects the cardinality of the new sets, such
that a set with higher cost turns into smaller sets of cost $1$.

We suppose that the costs of all sets are integer powers of $2$.
This can easily be achieved by rounding the costs to powers of $2$
(losing a factor of $2$ in the lower bound), and then scaling the
instance (both delays and buying costs) by the inverse of the lowest
cost.

Denote by $C=2^{M}$ the largest cost in the instance. The universe
of the unweighted instance is the following:
\begin{itemize}
	\item For each element $e$ in the original instance, we have $C$ elements
	in the unweighted instance, denoted by $e_{0},...,e_{C-1}$.
	\item For each set $S$, we have $c(S)$ sets in the unweighted instance,
	labeled $S_{0},...,S_{c(S)-1}$.
	\item We have that $e_{i}\in S_{j}$ if and only if both $e\in S$ and $i\equiv j\,\mod\,c(S)$.
\end{itemize}
Whenever a request $q_{j}$ arrives in the original instance on an
element $e$ with delay function $d_{j}(t)$, $C$ requests $q_{j,0},...,q_{j,C-1}$
arrive in the unweighted instance on the elements $e_{0},...,e_{C-1}$
respectively. For each $0\le l\le C-1$, the request $q_{j,l}$ has
the delay function $d_{j,l}(t)=\frac{d_{j}(t)}{C}$.

For the instance $I_{i}$ described above, we consider its unweighted
conversion, denoted by $I_{i}^{\prime}$. Any fractional online algorithm
for $I_{i}^{\prime}$ can be converted to a fractional online algorithm
for $I_{i}$ with a cost which is at most that of the original algorithm.
This is done by setting the buying function of a set $S$ in $I_{i}$
to the average of the buying functions of $S_{0},...,S_{c(S)-1}$.

In addition, the integral optimum described in the analysis of $I_{i}$
can be modified to an integral optimum for $I_{i}^{\prime}$ with
identical cost. This is by converting each buying of the set $S$
in $I_{i}$ to buying the sets $S_{0},...,S_{c(S)-1}$ in $I_{i}^{\prime}$.

The aforementioned facts imply that any fractional algorithm is $\Omega(\sqrt{i})$
competitive on $I_{i}^{\prime}$. Note that the parameter $k$ is
the same for $I_{i}$ and $I_{i}^{\prime}$, implying $\Omega(\sqrt{\log k})$-competitiveness
on $I_{i}^{\prime}$. In addition, denoting by $n^{\prime}$ the number
of elements in $I_{i}^{\prime}$, we have that $n^{\prime}=C\cdot n$.
Observing the construction of $I_{i}$, we have that $n=3^{i}$ and
$C\le2^{i}$ (Using the fact that $(1+\alpha_{j})\le2$ for any $j$).
Therefore, $\log n^{\prime}\le\log(6^{i})$, yielding that $i=\Omega(\log n^{\prime})$,
and a $\Omega(\sqrt{\log n^{\prime}})$ lower bound on competitiveness
for $I_{i}^{\prime}$.

\section{\label{sec:VertexCover}Vertex Cover with Delay}

In this section, we show a $3$-competitive deterministic algorithm
for VCD. Recall that VCD is a special case of SCD with $k=2$, where
$k$ is the maximum number of sets to which an element can belong.
In fact, we show a $(k+1)$-competitive deterministic algorithm for
SCD, which is therefore $3$-competitive for VCD. Recall that since
the TCP acknowledgment problem is a special case of VCD with a single
edge, the lower bound of $2$-competitiveness for any deterministic
algorithm on the TCP acknowledgment problem (shown in \cite{DBLP:conf/stoc/DoolyGS98})
applies to VCD as well.

The $(k+1)$-competitive algorithm for SCD, $\ON$, is as follows.
\begin{enumerate}
	\item For every set $S$, maintain a counter $z(S)$ of the total delay
	incurred by $\ON$ over requests on elements in $S$ (all $z(S)$ are
	initialized to $0$).
	\item If for any $S$, we have that $z(S)=c(S)$:
	\begin{enumerate}
		\item Buy $S$.
		\item Zero the counter $z(S)$.
	\end{enumerate}
\end{enumerate}
We denote by $z(S,t)$ the value of $z(S)$ at time $t$. We prove
the following theorem.
\begin{thm}
	\label{thm:VCD_Competitive}The algorithm $\ON$ for SCD has a competitive
	ratio of $k+1$. In particular, $\ON$ is $3$-competitive for VCD.
\end{thm}

\begin{lem}
	The cost of the algorithm is at most $k+1$ times its delay cost.
\end{lem}

\begin{proof}
	It is sufficient to bound the buying cost in terms of the delay cost.
	For each purchase of a set $S$, $z(S)$ must increase from $0$ to
	$c(S)$. A delay for a request contributes to the increase of at most
	$k$ counters. Thus, the buying cost is at most $k$ times the delay
	cost.
\end{proof}
We are left to bound the delay cost of the algorithm by the adversary's
cost.
\begin{lem}
	\label{lem:VCD_C1}For any set $S$, let $T$ be a subset of the requests
	on elements of $S$ such that all requests of $T$ are unserved at
	time $t$. Then we have $\sum_{j|q_{j}\in T}\int_{t}^{\infty}d_{j}^{\ON}(t^{\prime})\diff{t^{\prime}}\leq c(S)$.
\end{lem}

\begin{proof}
	Denote by $\hat{t}$ the first time in which all requests in $T$
	are served. We have that
	\[
	\sum_{j|q_{j}\in T}\int_{t}^{\infty}d_{j}^{\ON}(t^{\prime})\diff{t^{\prime}}=\sum_{j|q_{j}\in T}\int_{t}^{\hat{t}}d_{j}^{\ON}(t^{\prime})\diff{t^{\prime}}
	\]
	At time $t$, we have $z(S,t)\ge0$. Observe that the algorithm never
	bought $S$ in the time interval $[t,\hat{t})$. Thus, at any time
	$t^{\prime\prime}\in[t,\hat{t})$ we have that
	\[
	z(S,t^{\prime\prime})=z(S)+\sum_{j|q_{j}\in T}\int_{t}^{t^{\prime\prime}}d_{j}^{\ON}(t^{\prime})\diff{t^{\prime}}
	\]
	Observe that $z(S,t^{\prime\prime})<c(S)$, otherwise the algorithm
	would have bought $S$ at $t^{\prime}$, serving all requests in $T$,
	in contradiction to the definition of $\hat{t}$. Therefore $\sum_{j|q_{j}\in T}\int_{t}^{t^{\prime\prime}}d_{j}^{\ON}(t^{\prime})\diff{t^{\prime}}<c(S)$.
	The claim follows as $t^{\prime\prime}$ approaches $\hat{t}$.
\end{proof}
\begin{lem}
	\label{lem:VCD_DelayCharge}The delay cost of the algorithm is at
	most the adversary's cost.
\end{lem}

\begin{proof}
	We construct a solution to the dual LP from section \ref{sec:FractionalAlgorithm},
	with a goal function which is the delay cost of the algorithm. This
	charges the delay cost of the algorithm to the fractional optimum,
	and thus to the integer optimum as well.
	
	Specifically, we set $y_{j}(t)=d_{j}^{\ON}(t)$ for every $j,t$. Obviously,
	the \ref{eq:C2} constraints hold. In order to show that the \ref{eq:C1}
	constraint for a set $S_{i}$ and a time $t$ holds, observe that
	any request $q_{j}\in S_{i}$ served in $\ON$ before time $t$ has
	$d_{j}^{\ON}(t^{\prime})=0$ for all $t^{\prime}\ge t$. Using Lemma
	\ref{lem:VCD_C1} for the requests unserved at $t$ concludes the
	proof.
\end{proof}
\begin{proof}
	(of theorem \ref{thm:VCD_Competitive}) The proof of the theorem results
	directly from lemmas \ref{lem:VCD_C1} and \ref{lem:VCD_DelayCharge}.
\end{proof}
Note that this algorithm's competitive ratio is indeed as bad as $k+1$.
Consider, for example, a single request in $k$ sets with unit costs,
which the optimum solves with cost $1$ and the algorithm has cost
$k+1$.

\bibliographystyle{plain}
\bibliography{SetCoverWithDelay_short}

\newpage

\appendix

\section*{Appendix}

\section{Randomized Algorithm for SCD by Rounding -- Proofs}

\subsection{\label{subsecAppendix:RequestRoundingProof}Proof of Theorem \ref{thm:RequestRandomizedCompetitive}}

We split the buying cost of $\ONR$, $\text{Cost}_{\ONR}^{p}$, into
two parts: the total ``type a'' buying cost, denoted $\text{Cost}_{\ONR}^{a}$,
and the total ``type b'' buying cost, denoted $\text{Cost}_{\ONR}^{b}$.
\begin{lem}
\label{lem:TypeAGood}$\E[\text{Cost}_{\ONR}^{a}]\le4\ln N\cdot\text{Cost}_{\ONF}$.
\end{lem}

\begin{proof}
To show the lemma, fix any set $S_{i}$. We observe the values chosen
for $\Lambda_{i}$ in the algorithm as a sequence $(\Lambda_{i}^{l})_{l=1}^{\infty}$
of independent random variables, taken uniformly from $[0,\frac{1}{2\ln N}]$.
Whenever the algorithm buys $S_{i}$ via ``type a'', it reveals
the next element of the sequence. Denoting by $s$ the number of times
$S_{i}$ is ``type a'' bought, we have that for every $l$ the indicator
variable $1_{s+1\ge l}$ and $\Lambda_{i}^{l}$ are independent (the
value of $\Lambda_{i}^{l}$ does not affect whether the algorithm
reveals it). Since the elements of the sequence are equidistributed,
we can use the general version of Wald's equation to obtain:
\begin{equation}
\E\left[\sum_{l=1}^{s+1}\Lambda_{i}^{l}\right]=\E[s+1]\cdot\E[\Lambda_{i}^{1}]\ge\frac{\E[s]}{4\ln N}+\frac{1}{4\ln N}\label{eq:Expecto}
\end{equation}

Denoting by $t^{\prime}$ the last time that $S_{i}$ was ``type
a'' bought, we also know that

\[
\sum_{l=1}^{s}\Lambda_{i}^{l}=\int_{0}^{t^{\prime}}x_{i}(t)\diff{t}\le\int_{0}^{\infty}x_{i}(t)\diff{t}
\]

since all revealed thresholds but $\Lambda_{i}^{s+1}$ have been surpassed
by $x_{i}(t)$. Therefore
\begin{align*}
\E\left[\sum_{l=1}^{s+1}\Lambda_{i}^{l}\right] & =\E\left[\sum_{l=1}^{s}\Lambda_{i}^{l}\right]+\E[\Lambda_{i}^{s+1}]\\
 & \le\int_{0}^{\infty}x_{i}(t)\diff{t}+\frac{1}{4\ln N}
\end{align*}
Combining this with equation \ref{eq:Expecto} yields
\[
\E[s]\le4\ln N\cdot\int_{0}^{\infty}x_{i}(t)\diff{t}
\]
and thus
\[
\E[c(S_{i})\cdot s]\le4\ln N\cdot c(S_{i})\cdot\int_{0}^{\infty}x_{i}(t)\diff{t}
\]

Note that the total ``type a'' buying cost of $S_{i}$ is $c(S_{i})\cdot s$,
while the buying cost of $S_{i}$ in $\ONF$ is $c(S_{i})\cdot\int_{0}^{\infty}x_{i}(t)\diff{t}$.
Summing the previous inequality over all $S_{i}$ therefore yields
the lemma.
\end{proof}
\begin{lem}
\label{lem:DelayGood}$\text{Cost}_{\ONR}^{d}\le2\cdot\text{Cost}_{\ONF}$.
\end{lem}

\begin{proof}
Due to the ``type b'' buying, if a request $q_{j}$ is pending in
$\ONR$ at time $t$, we have that $\gamma_{j}(t)\le\frac{1}{2}$.
Thus $d_{j}^{\ONF}(t)\ge\frac{1}{2}\cdot d_{j}(t)$, and therefore
the $\ONF$ always incurs at least half the delay cost of $\ONR$. This
yields the lemma.
\end{proof}
It remains to bound the total ``type b'' buying. For any request
$q_{j}$ and time $t\ge r_{j}$, we define the event $A_{j}^{t}$,
which is the event that $q_{j}$ has not been served in $\ONR$ by
time $t$.
\begin{lem}
\label{lem:RequestProbabilityLow}For any request $q_{j}$, with $A_{j}^{t}$
as defined above, we have
\[
Pr(A_{j}^{t_{j}})\le\frac{1}{N}
\]
\end{lem}

\begin{proof}
For $S_{i}$ a set and $I=[t_{1},t_{2})$ a time interval, denote
by $A_{i}^{I}$ the event that $i$ has not been bought by ``type
a'' buying in $I$. Denote by $\Lambda_{i}^{l}$ the current threshold
for $S_{i}$ at time $t_{1}$, and denote by $t^{\prime}\le t_{1}$
the time the threshold was set. We have that
\[
Pr(A_{i}^{I})=Pr\left(\int_{t^{\prime}}^{t_{2}}x_{i}(t)\diff{t}\le\Lambda_{i}^{l}\,\middle|\,\int_{t^{\prime}}^{t_{1}}x_{i}(t)\diff{t}\le\Lambda_{i}^{l}\right)
\]
Fix $t^{\prime}$. Given that $\int_{t^{\prime}}^{t_{1}}x_{i}(t)\diff{t}\le\Lambda_{i}^{l}$,
we have that $\Lambda_{i}^{l}\sim U\left(\int_{t^{\prime}}^{t_{1}}x_{i}(t)\diff{t},\frac{1}{2\ln N}\right)$.
Defining $\Lambda=\Lambda_{i}^{l}-\int_{t^{\prime}}^{t_{2}}x_{i}(t)\diff{t}$,
we have $\Lambda\sim U\left(0,\frac{1}{2\ln N}-\int_{t^{\prime}}^{t_{1}}x_{i}(t)\diff{t}\right)$
and thus
\begin{align*}
Pr(A_{i}^{I}) & =Pr\left(\int_{t_{1}}^{t_{2}}x_{i}(t)\diff{t}\le U\left(0,\frac{1}{2\ln N}-\int_{t^{\prime}}^{t_{1}}x_{i}(t)\diff{t}\right)\right)\\
 & \le Pr\left(\int_{t_{1}}^{t_{2}}x_{i}(t)\diff{t}\le U\left(0,\frac{1}{2\ln N}\right)\right)\\
 & =\begin{cases}
1-2\ln N\cdot\int_{t_{1}}^{t_{2}}x_{i}(t)\diff{t} & \text{if }\int_{t_{1}}^{t_{2}}x_{i}(t)\diff{t}\le\frac{1}{2\ln N}\\
0 & \text{otherwise}
\end{cases}
\end{align*}

Note that for two distinct sets $S_{i_{1}},S_{i_{2}}$, the events
$A_{i_{1}}^{I_{1}}$ and $A_{i_{2}}^{I_{2}}$ are independent for
any two time intervals $I_{1},I_{2}$. We have that
\[
Pr(A_{j}^{t})\le Pr\left(\bigwedge_{i|q_{j}\in S_{i}}A_{i}^{[r_{j},t)}\right)=\prod_{i|q_{j}\in S_{i}}Pr(A_{i}^{[r_{j},t)})
\]

where the equality follows from the independence of the events. We
now analyze $A_{j}^{t_{j}}$. If there exists $i$ such that $q_{j}\in S_{i}$
and $\int_{r_{j}}^{t_{j}}x_{i}(t)\diff{t}>\frac{1}{2\ln N}$, then $Pr(A_{i}^{[r_{j},t_{j})})=0$
and thus $Pr(A_{j}^{t_{j}})=0$ and the proof is complete. Otherwise,
for all such $i$, we have that $Pr(A_{i}^{[r_{j},t_{j})})\le1-2\ln N\int_{r_{j}}^{t_{j}}x_{i}(t)\diff{t}$.
Denote $k_{j}=|\{i|q_{j}\in S_{i}\}|$. This implies
\begin{align*}
Pr(A_{j}^{t_{j}}) & \le\prod_{i|q_{j}\in S_{i}}\left(1-2\ln N\int_{r_{j}}^{t_{j}}x_{i}(t)\diff{t}\right)\\
 & \le\left(1-2\ln N\cdot\frac{\sum_{i|q_{j}\in S}\int_{r_{j}}^{t_{j}}x_{i}(t)\diff{t}}{k_{j}}\right)^{k_{j}}\\
 & \le\left(1-2\ln N\cdot\frac{1}{2k_{j}}\right)^{k_{j}}\\
 & =\left(1-\frac{\ln N}{k_{j}}\right)^{\frac{k_{j}}{\ln N}\cdot\ln N}\\
 & \le e^{-\ln N}=\frac{1}{N}
\end{align*}

where the second inequality follows from taking the arithmetic mean
of the factors and raising it to the power of their number. The third
inequality follows from the definition of $t_{j}$.
\end{proof}
\begin{cor}
\label{cor:TypeBGood}$\E[\text{Cost}_{\ONR}^{b}]\le 2\text{Cost}_{\ONF}$
\end{cor}

\begin{proof}
	First observe that if $\ONF$ covers less than half of a request $q_j$, then $q_j$ does not trigger any ``type b'' buying. Let $Q'$ be the subset of requests which are at least half-covered by $\ONF$ during the run of the algorithm. 
	We define $j^{\ast}=\arg\max_{j|q_j \in Q'}c(S_{i_{j}})$. We have that
	\begin{align*}
	\E[\text{Cost}_{\ONR}^{b}] & =\sum_{j|q_j \in Q'}c(S_{i_{j}})\cdot Pr(A_{j}^{t_{j}})\\
	 & \le\frac{1}{N}\sum_{j|q_j \in Q'}c(S_{i_{j}})\\
	 & \le\frac{1}{N}\sum_{j|q_j \in Q'}c(S_{i_{j^{\ast}}})=c(S_{i_{j^{\ast}}})
	\end{align*}
	where the first equality is due to linearity of expectation, the first
	inequality is due to Lemma \ref{lem:RequestProbabilityLow}. Now note
	that since $\ONF$ covered $q_{j^{\ast}}$ by a fraction of at least half, it paid a total cost of at least $\frac{c(S_{i_{j^{\ast}}})}{2}$. This concludes the proof.
\end{proof}
We now prove the main theorem.
\begin{proof}
(of Theorem \ref{thm:RequestRandomizedCompetitive}) Combining Lemmas
\ref{lem:TypeAGood} and \ref{lem:DelayGood} with Corollary \ref{cor:TypeBGood}
yields:
\begin{align*}
\E[\text{Cost}_{\ONR}] & =\E[\text{Cost}_{\ONR}^{p}+\text{Cost}_{\ONR}^{d}]\\
 & =\E[\text{Cost}_{\ONR}^{a}]+\E[\text{Cost}_{\ONR}^{b}]+\E[\text{Cost}_{\ONR}^{d}]\\
 & \le(4\ln N+4)\cdot\text{Cost}_{\ONF}=O(\log N)\cdot\text{Cost}_{\ONF}
\end{align*}

Since $\ONF$ is $O(\log k)$ competitive with respect to the fractional
optimum, we get that $\ONR$ is $O(\log N\cdot\log k)$ competitive
with respect to the fractional optimum, and in particular the integral
optimum.
\end{proof}

\subsection{\label{subsecAppendix:ElementRoundingProof}Improved Rounding and Proof of Theorem \ref{thm:ElementRandomizedCompetitive}}

In this subsection, we show how to modify the $O(\log k\cdot\log N)$-competitive
randomized-rounding algorithm shown in Section \ref{sec:RoundingAlgo}
to yield a $O(\log k\cdot\log n)$-competitive randomized algorithm.
The intuition behind the modifications is removing the dependency
on the number of requests by aggregating requests on the same element.
Specifically, we discretize time into intervals, such that requests
on the same element that arrive in the same interval are aggregated.
Instead of having a threshold time for ``type b'' buying for every
request, we have a threshold time for every interval.
\begin{defn}
	For every element $e$, we define \emph{threshold times}, spaced by
	$\ONF$ buying fractions of sets containing $e$ which sum to a constant. Formally,
	for every element $e$, we define the threshold time $t_{l}^{e}$
	for $l\in\mathbb{N}$ to be the first time for which $\int_{0}^{t_{l}^{e}}\left(\sum_{i|e\in S_{i}}x_{i}(t)\right)\diff{t}=\frac{l}{4}$.
\end{defn}

Denote by $s_{e}$ the index of the last threshold time for $e$.
Using the definition of $t_{s_{e}}^{e}$, we have
\begin{equation}
\int_{0}^{\infty}\left(\sum_{i|e\in S_{i}}x_{i}(t)\right)\diff{t}\ge\frac{s_{e}}{4}\label{eq:BuyingAtLeastLength}
\end{equation}

For simplicity, we denote $t_{0}^{e}=0$. Define $R_{l}^{e}$ for
$0\le l\le s_{e}-1$ to be the set of requests released on $e$ in
the interval $[t_{l}^{e},t_{l+1}^{e})$. Note that no request is released
outside of some $R_{l}^{e}$ -- if a request is released on element
$e$ after $t_{s_{e}}^{e}$, it would require set buying by $\ONF$
which would create three new threshold times, in contradiction to
$s_{e}$'s definition. For the same reason, $R_{s_{e}-1},R_{s_{e}-2}$
are empty.

If at time $t$ all the requests of $R_{l}^{e}$ have been served,
we say that $R_{l}^{e}$ has been served. Otherwise, $R_{l}^{e}$
is unserved at time $t$.

We modify the $O(\log k\cdot\log N)$-competitive algorithm of Section \ref{sec:RoundingAlgo} as follows:
\begin{enumerate}
	\item The \textquotedbl type a\textquotedbl{} thresholds $\Lambda_{i}$
	are now drawn from $U\left(0,\frac{1}{2\ln n}\right)$ (using $n$
	instead of $N$).
	\item \label{enu:Step2-1}''Type b'' buying is changed to the following
	rule -- for every element $e$ and every $l$, if $R_{l}^{e}$ remains
	unserved until $t_{l+3}^e$, we buy $S_{e}$.
\end{enumerate}
Note that $t_{l+3}^e$ in (\ref{enu:Step2-1}) is well defined since
$R_{s_{e}-1},R_{s_{e}-2}$ are empty.

We now prove Theorem \ref{thm:ElementRandomizedCompetitive}.

As in Appendix \ref{subsecAppendix:RequestRoundingProof}, we define
$\text{Cost}_{\ONR}^{a}$ and $\text{Cost}_{\ONR}^{b}$ to be the ``type
a'' buying cost and the ``type b'' buying cost of the algorithm,
respectively.
\begin{lem}
\label{lem:ElementTypeALow}$\E[\text{Cost}_{\ONR}^{a}]\le4\ln n\cdot\text{Cost}_{\ONF}$
\end{lem}

\begin{proof}
The proof is identical to that of Lemma \ref{lem:TypeAGood}.
\end{proof}
For every $R_{l}^{e}$, we also define $\Gamma_{l}^{e}(t)$ for $t\ge t_{l+1}^{e}$,
which is the fraction of $e$ covered by $\ONF$ from time $t_{l+1}^{e}$:
\[
\Gamma_{l}^{e}(t)=\sum_{i|e\in S_{i}}\int_{t_{l+1}^{e}}^{t}x_{i}(t^{\prime})\diff{t^{\prime}}
\]

\begin{prop}
\label{prop:RatherUntouched}For $q_{j}\in R_{l}^{e}$, we have that
$\gamma_{j}(t_{l+1})\le\frac{1}{4}$.
\end{prop}

\begin{proof}
Otherwise, the fractional algorithm has bought a total fraction of
more than $\frac{1}{4}$ of sets containing $e$ in $[t_{l}^{e},t_{l+1}^{e})$,
a contradiction to the definition of threshold times.
\end{proof}
\begin{lem}
If a request $q_{j}$ is pending in the randomized algorithm at time
$t$, then
\[
\gamma_{j}(t)\le\frac{3}{4}
\]
\end{lem}

\begin{proof}
Choose $R_{l}^{e}$ such that $q_{j}\in R_{l}^{e}$. If $t\le t_{l+1}^{e}$,
the lemma results from Proposition \ref{prop:RatherUntouched} and
we're done.

Otherwise, $t>t_{l+1}^{e}$ and therefore $\gamma_{j}(t)=\gamma_{j}(t_{l+1}^{e})+\Gamma_{l}^{e}(t)$.
Since $q_{j}$ is pending at $t$, we have that $R_{l}^{e}$ is unserved
at $t$. This implies that $t\le t_{l+3}^{e}$. From the definition
of threshold times, $\sum_{i|e\in S_{i}}\int_{t_{l+1}^{e}}^{t_{l+3}^{e}}x_{i}(t^{\prime})\diff{t^{\prime}}\le\frac{1}{2}$
and thus $\Gamma_{l}^{e}(t)\le\frac{1}{2}$. Therefore
\[
\gamma_{j}(t)=\gamma_{j}(t_{l+1}^{e})+\Gamma_{l}^{e}(t)\le\frac{1}{4}+\frac{1}{2}=\frac{3}{4}
\]
where the inequality uses Proposition \ref{prop:RatherUntouched}.
\end{proof}
\begin{cor}
\label{cor:ElementDelayLow}$\text{Cost}_{\ONR}^{d}\le4\cdot\text{Cost}_{\ONF}$.
\end{cor}

\begin{proof}
Immediate from the previous lemma.
\end{proof}
It remains to bound the expected \textquotedbl type b\textquotedbl{}
buying.
\begin{prop}
\label{prop:SingleExpectancy}The probability that $R_{l}^{e}$ triggers
``type b'' buying is at most $\frac{1}{n}$.
\end{prop}

\begin{proof}
It is enough to show that the probability that the algorithm does
not perform ``type a'' buying of a set containing $e$ during $[t_{l+1},t_{l+3})$
is at most $\frac{1}{n}$. Showing this is identical to the proof
of Lemma \ref{lem:RequestProbabilityLow}.
\end{proof}
\begin{prop}
\label{prop:FractionalHigh}The total cost of $\ONF$ is at least $\frac{1}{4}\cdot s_{e}\cdot c(S_{e})$,
for any element $e$.
\end{prop}

\begin{proof}
From Equation \ref{eq:BuyingAtLeastLength}, we have that $\ONF$ buys
at least $\frac{s_{e}}{4}$ fraction of sets containing $e$. Since
$S_{e}$ is the least expensive set containing $e$, $\ONF$ must have
payed a buying cost of at least $\frac{1}{4}\cdot s_{e}\cdot c(S_{e})$.
From the definition of threshold times, and the definition of $S_{e}$
as the least expensive set containing $e$.
\end{proof}
\begin{prop}
\label{lem:RandomizedLow}For every element $e$, the total expected
``type b'' buying cost for that element is at most $\frac{1}{n}\cdot s_{e}\cdot c(S_{e})$
\end{prop}

\begin{proof}
Let $X_{l}^{e}$ be the indicator random variable of $R_{l}^{e}$
being ``type b'' bought. The lemma results directly from linearity
of expectation and Proposition \ref{prop:SingleExpectancy}.
\end{proof}
\begin{lem}
\label{lem:ElementTypeBLow}$\E[\text{Cost}_{\ONR}^{b}]\le4\cdot\text{Cost}_{\ONF}$.
\end{lem}

\begin{proof}
We fix $e^{\ast}=\arg\max_{e}(s_{e}\cdot c(S_{e}))$. Proposition
\ref{lem:RandomizedLow} implies that the expected ``type b'' buying
cost is at most:
\[
\sum_{e}\frac{1}{n}\cdot s_{e}\cdot c(S_{e})\le\frac{1}{n}\sum_{e}s_{e^{\ast}}\cdot c(S_{e^{\ast}})=s_{e^{\ast}}\cdot c(S_{e^{\ast}})\le4\cdot\text{Cost}_{\ONF}
\]
where the first inequality is from the definition of $e^{\ast}$,
and the last inequality is from Proposition \ref{prop:FractionalHigh}.
This concludes the proof.
\end{proof}
We now prove the main theorem.
\begin{proof}
(of Theorem \ref{thm:ElementRandomizedCompetitive}) Using Lemmas
\ref{lem:ElementTypeALow}, \ref{cor:ElementDelayLow} and \ref{lem:ElementTypeBLow},
we get:
\[
\E[\text{Cost}_{\ONR}^{a}+\text{Cost}_{\ONR}^{b}+\text{Cost}_{\ONR}^{d}]\le(4\ln n+8)\cdot\text{Cost}_{\ONF}
\]
which proves the theorem.
\end{proof}

\section{\label{sec:NonclairvoyantIsOptimal2}Reduction from Online Set Cover to Set Cover with Delay}
In this section, we show an online reduction from the online set cover problem of \cite{DBLP:journals/siamcomp/AlonAABN09} (denoted OSC) to the non-clairvoyant SCD problem of this paper which preserves the running time and competitive ratio. A lower bound of $\Omega(\log n \log m)$ on the competitiveness of any polynomial-time, randomized algorithm for OSC (conditioned on $NP\not\subseteq BPP$) is given in \cite{KormanSetCoverRandomization}; this implies that our $O(\log n \log m)$-competitive algorithm is optimal for non-clairvoyant SCD.

The reduction works by creating an instance of set cover with deadlines, a special case of SCD. Since we consider non-clairvoyant SCD, the algorithm is not aware of the future delay of a request before it arrives. In the case of deadlines, this translates to the input "announcing" the deadline of a previously-released request $q$ at some point in time $t$, which forces the algorithm to immediately serve $q$ if it is still pending. This allows the adversary to refrain from committing to a specific deadline upon the release of a request -- a crucial property in the construction.

We proceed to show the online reduction. Suppose there exists an $\alpha$-competitive algorithm $\SCDALG$ for set cover with delay. We construct an $\alpha$-competitive algorithm for $\OSC$ (which uses $\SCDALG$ as a black box), which operates on the $\OSC$ request sequence $e_1,\cdots, e_l$. The algorithm is operates in the following way:
\begin{enumerate}
	\item Initially, create a $\SCD$ input for $\SCDALG$ starting at time $0$. Release a set of $n$ requests, one on each element, at time $0$. Let $S\gets \emptyset$ be the collection of sets bought thus far.
	
	\item For $i$ from $1$ to $l$:
	\begin{enumerate}
		\item Receive the next requested element $e_i$. 
		
		\item Advance $\SCDALG$'s time to $i$, and announce the deadline of the request on $e_i$. 
		
		\item Let $S_i$ be the collection of sets bought by $\SCDALG$ upon the deadline of the request on $q_i$.
		
		\item Buy the sets in $S_i$, setting $S\gets S\cup S_i$.
	\end{enumerate}
\end{enumerate}

To analyze this reduction, consider the final $\SCD$ input as terminating after time $l$ -- that is, the remaining $n-l$ requests released at time $0$ never reach their deadline, and thus do not have to be served. 

\paragraph*{Feasibility.} We need to show that the $\OSC$ solution indeed covers every element in the request sequence upon its arrival. This holds since the algorithm holds after the $i$'th request the union of the sets bought by $\SCDALG$ until time $i$ -- which must include the set which covers the $i$'th request.

\paragraph*{Cost.} The cost of $\SCDALG$ for the given input is at most $\alpha$ times the optimal cost for the input. Now, observe that buying the optimal set cover for $\{e_1,\cdots, e_l\}$ immediately after time $0$ is a feasible $\SCD$ solution for the given input (the cost of which is exactly the cost of the optimal $\OSC$ solution). Also note that the cost of the $\OSC$ algorithm is at most the cost of  $\SCDALG$. This implies that the $\OSC$ algorithm is also $\alpha$-competitive.

Clearly, the asymptotic running time of the $\OSC$ algorithm is exactly that of $\SCDALG$.

%
%

\end{document}